\documentclass[10pt,journal,compsoc]{IEEEtran}
\pdfoutput=1 

\usepackage{makecell}
\usepackage{tabularx}

\newcolumntype{Y}{>{\centering\arraybackslash}X}
\usepackage{hyperref}
\usepackage{color,overpic,cite}
\usepackage{xcolor}
\usepackage{url}
\usepackage{tcolorbox}
\usepackage[ruled,commentsnumbered]{algorithm2e} 

\newcommand{\eat}[1]{{}}

\newcommand{\argmin}{\arg\min}

\def\p{\widetilde}
\def\grad{\nabla}

\def\ll{\lambda}

\def\sumT{\sum_{t=1}^{T}}

\usepackage{amsmath,amsfonts,amssymb,amsthm,bbm,bm,comment,multirow,subfigure}
\DeclareMathOperator*{\argmax}{argmax}

\newtheorem{theorem}{Theorem}
\newtheorem{lemma}{Lemma}

\usepackage{enumitem}

\usepackage{mdframed}
\usepackage{bm}

\title{Online Caching with no Regret: Optimistic Learning via Recommendations}

\begin{document}

\author{Naram~Mhaisen,
        George~Iosifidis,
        and~Douglas~Leith
\IEEEcompsocitemizethanks{\IEEEcompsocthanksitem N. Mhaisen and G. Iosifidis are with the Department of Software Technology, Delft University of Technology, Netherlands.\protect\\
E-mails: \{n.mhaisen, g.iosifidis\}@tudelft.nl
\IEEEcompsocthanksitem D. Leith is with the School of Computer Science and Statistics, Trinity College Dublin, Ireland. E-mail: doug.leith@tcd.ie}
\thanks{A preliminary version of this work appears in the proceedings of IFIP Networking 2022 \cite{ocol}. The current version expands the content by including a new theorem (Th. \ref{th:regret1}), which improves the regret bound; designing an optimistic meta-learning policy that accommodates multiple predictors; and expanding the numerical analysis using additional datasets and experiments.}}

\IEEEtitleabstractindextext{%

\begin{abstract}
The design of effective online caching policies is an increasingly important problem for content distribution networks, online social networks and edge computing services, among other areas. This paper proposes a new algorithmic toolbox for tackling this problem through the lens of \emph{optimistic} online learning. We build upon the Follow-the-Regularized-Leader (FTRL) framework, which is developed further here to include predictions for the file requests, and we design online caching algorithms for bipartite networks with pre-reserved or dynamic storage subject to time-average budget constraints. The predictions are provided by a content recommendation system that influences the users viewing activity and hence can naturally reduce the caching network's uncertainty about future requests. We also extend the framework to learn and utilize the best request predictor in cases where many are available. We prove that the proposed {optimistic} learning caching policies can achieve \emph{sub-zero} performance loss (regret) for perfect predictions, and maintain the sub-linear regret bound $O(\sqrt T)$, which is the best achievable bound for policies that do not use predictions, even for arbitrary-bad predictions. The performance of the proposed algorithms is evaluated with detailed trace-driven numerical tests.
\end{abstract}

}

\maketitle

\section{Introduction}

\subsection{Motivation and Background}

The quest for efficient data caching policies spans more than 50 years and remains today  one of the most important research areas for wireless and wired communication systems \cite{paschos-jsac}. Caching was first studied in computer systems where the aim was to decide which files to store in fast-accessible memory segments (\emph{paging}) \cite{Belady66}. Its scope was later expanded due to the explosion of Internet web traffic \cite{aggarwal-www} and the advent of content distribution networks (CDNs) \cite{RossComCom02}, and was recently revisited as a technique to improve the operation of wireless networks through edge caches \cite{femtocaching} and on-device caching \cite{femtocaching_d2d}. A common challenge in these systems is to design an online policy that decides which files to store at a cache, without knowing the future file requests, so as to maximize the cache \emph{hits} or some other  cache-related performance metric.

There is a range of online caching policies that tackle this problem under different assumptions on the request arrivals. Policies such as the LFU and LRU are widely-deployed, yet their performance deteriorates when the file popularity is non-stationary, i.e., the requests are drawn from a time-varying probability distribution \cite{Sleator85, lru-sigmetrics08, lfu-sigmetrics99}. This motivated modeling non-stationary request patterns \cite{snm, kauffman} and optimizing accordingly the caching decisions \cite{mathieu, Elayoubi2015}. Another line of work relies on techniques such as reinforcement learning to estimate the request probabilities and make caching decisions accordingly \cite{gunduz-reinforcement, giannakis-q-learning}; but typically these solutions either do not offer optimality bounds, or do not scale due to having the library size in their bounds.

Caching was studied within the framework of online learning in \cite{geulen2010regret} for a single-cache system; and in its more general form recently in \cite{paschos-infocom19} that proposed an online gradient descent (OGD) caching policy. Interesting follow-up works include sub-modular policies \cite{Li-online-2021}, online mirror-descent policies \cite{stratis-2020}, and the characterization of their performance limits \cite{abhishek-sigm20, abhishek_nurips}. The advantage of these online learning-based caching policies is that they are scalable, do not require training data, and their performance bounds are \emph{robust} to any possible request pattern, even when the requests are generated by an adversary that aims to degrade the caching operation.

Nevertheless, an aspect that remains hitherto unexplored is whether predictions about future requests can improve the performance of such learning-based caching policies \emph{without sacrificing their robustness}. This is important in modern caching systems where often the users receive content viewing recommendations from a recommendation system (\emph{rec-sys}). For instance, recommendations are a standard feature in streaming platforms such as YouTube and Netflix \cite{netflix}; but also in online social network platforms such as Facebook and Twitter, which moderate the users' viewing feeds \cite{recommend2}. Not surprisingly, the interplay between recommendations and caching attracted recent attention and prior works aimed to increase the caching hits or reduce routing costs, by either recommending already-cached files to users, or through the joint optimization of caching and recommendation decisions \cite{jordan-tmc, quek-twc21, sheng-twc21, 9528062, 9681354}. These important works, however, consider \emph{static} caching models and require knowing in advance the users' expected requests and their propensity to follow the recommendations.

Changing vantage point, one can observe that since recommendations bias the users towards viewing certain contents, they can effectively serve as predictions of the forthcoming requests. This prediction information, if properly employed, can hugely improve the efficacy of \emph{dynamic} caching policies, transforming their design from an online learning to an online optimization problem. Nevertheless, the caching policy needs to adapt to the accuracy of recommendations (i.e., of the predictions) and the users propensity to follow them -- which is typically unknown and potentially time-varying. Otherwise, the caching performance might as well deteriorate by following these misleading \emph{hints} about the future requests. The goal of this work is to tackle exactly this challenging new problem and answer the question: \emph{Can we leverage untrusted predictions in caching systems?}. We answer this question in the affirmative by  \emph{proposing online learning-based caching policies that utilize predictions (of unknown quality) to boost performance, if those predictions are accurate, while still maintaining robust performance bounds otherwise}.

\subsection{Methodology and Contributions}

Our approach is based on the theory of Online Convex Optimization (OCO) that was introduced in \cite{zinkevich2003online} and has since been applied in several decision problems \cite{hazan-book}. The basic premise of OCO is that a learner (here the caching system) selects in each slot $t$ a decision vector $x_t$ from a convex set $ \mathcal X$, without knowing the $t$-slot convex performance function $f_t(x)$, that changes with time. The learner's goal is to minimize the growth rate of \emph{regret} $R_T\!=\!\sumT f_t(x^\star)\!-\!f_t(x_t)$, where $x^\star\!=\!\arg\max_{x\in \mathcal X} \sumT f_t(x)$ is the benchmark solution designed with hindsight, i.e., with access to the entire sequence of future functions $\{f_t\}_{t=1}^T$. The online caching problem fits squarely in this setup, where $f_t(x)$ depends on the users requests and is unknown when the caching is decided. And previous works \cite{paschos-infocom19, Li-online-2021, stratis-2020, abhishek-sigm20} have proved that OCO-based caching policies achieve $R_T\!=\!O(\sqrt T)$, thus ensuring asymptotically zero average regret: $\lim_{T\rightarrow \infty} R_T/T\!=\!0$. 

Different from these important studies, we extend the learning model to include predictions that are available through the content recommendations. Improving the regret of learning policies via {predictions} is a relatively new area in machine learning research. For instance, \cite{antoniadis_20} focuses on the \emph{competitive-ratio} metric and developed algorithms that use untrusted predictions while maintaining worst-case performance bounds; while \cite{Lykouris-ML} applied similar ideas to the paging problem. However, it was shown in \cite{pmlr-v30-Andrew13} that such competitive-ratio algorithms cannot ensure sublinear regret, which is the performance criterion we employ here, in line with all recent works \cite{paschos-infocom19, Li-online-2021, stratis-2020, abhishek-sigm20, abhishek_nurips}.  

For regret-minimization with predictions, \cite{dekel-nips17} used predictions $\p c_t$ for the function gradient $c_t\!=\!\grad f_t(x_t)$ with guaranteed quality, i.e., $c_t^\top \p c_t \!\ge\! a \|c_t\|^2$, to reduce $R_T$ from $O(\sqrt T)$ to $O(\log T)$; and \cite{google-2020} enhanced this result by allowing some predictions to fail the quality condition. A different line of works uses \emph{regularizing functions} which enable the learner to adapt to the predictions' quality \cite{sridharan-nips2013}, \cite{mohri-aistats2016}. This idea is more promising for the caching problem where the recommendations might be inaccurate, or followed by the users for only arbitrary time windows; thus, we used it as starting point to develop our caching learning frameworks.

In specific, our approach relies on the Follow-The-Regularized-Leader (FTRL) algorithm \cite{shalev-ftrl} which we extend with predictions that offer \emph{optimism} by reducing the uncertainty about the next-slot functions. We study different versions of the caching problem. First, we design a policy (OFTRL) for the bipartite caching model \cite{femtocaching}, which generalizes the standard single cache case \cite{geulen2010regret, Lykouris-ML}. Theorem \ref{th:regret1} proves that $R_T$ is proportional to prediction errors ($\|c_t\!-\p c_t\|^2, \forall t$) diminishing to zero for perfect predictions; while still meeting the best achievable bound $O(\sqrt T)$ for the regular OCO setup (i.e., without using predictions)  \cite{abhishek-sigm20} even if all predictions fail. We continue with the \emph{elastic} caching problem, where the system resizes the used caches at each slot based, e.g., on volatile storage leasing costs \cite{giannakis-elasticJSAC19, jungho-wiopt, akamai}. The aim is to maximize the caching utility subject to a time-average budget constraint. This places the problem in the realm of constrained-OCO \cite{giannakis-TSP17, paschos-icml, victor, johansson-TSP2020}. Using a new saddle point analysis with predictions, we prove Theorem \ref{th:regret-e} which reveals how $R_T^{(e)}$ and the budget violation $V_T^{(e)}$ depend on the cache sizes and prediction errors, and how one can prioritize one metric over the other while achieving sublinear growth rates for both. 

The above algorithms utilize the rec-sys as the only source to predict the next time-slot cost $\tilde{c}_{t+1}$. In many cases, however, content providers might have access to \emph{multiple} such sources. For example,  a statistical user profiling model, a deep learning-based predictive model for content request in the next time slot \cite{9468921}, or even another rec-sys, see \cite{mult_recsys} and follow-up works. Those sources can be used to obtain multiple, and possibly contradicting, predictions. Our final contribution is, therefore, a \emph{meta-learning} caching framework that utilizes predictions from multiple sources to achieve the same performance as a caching system that used the best such source to start with. We show that the regret in the case of multiple sources can be strictly negative depending on the request sequence and the existence of a high-accuracy predictor. At the same time, the meta-learning caching framework maintains sublinear regret when all predictors fail. Finally, we show that this framework can also be applied in cases with single rec-sys, and discuss its pros and cons compared to the proposed regularization-based optimistic caching solutions.

In summary, the contributions of this work can be grouped as follows: 

$\bullet$ Introduces an online learning framework for bipartite and elastic caching networks that leverages predictions to achieve a regret that is upper-bounded by \emph{zero} for perfect recommendations and sub-linear $O(\sqrt T)$ for arbitrary bad recommendations. The results are based on a new analysis technique that improves the bounds by a factor of $\sqrt{2}$ compared to the state-of-art optimistic-regret bounds \cite{ mohri-aistats2016}, which we used in our recent work \cite{ocol}.

$\bullet$ Introduces a meta-learning framework that can utilize predictions from \emph{multiple} sources and jointly learns which of them to use, if any at all, and the optimal caching decisions. This framework achieves \emph{negative} regret in the best case, and sublinear regret in the worst case.

$\bullet$ Evaluates the policies using various request models and real datasets \cite{zink2008watch, mlds} and compares them with \emph{(i)} the best in hindsight benchmark; and \emph{(ii)} the online gradient descent policy, which is known to be regret-optimal and outperforms other caching policies \cite{paschos-infocom19, abhishek-sigm20}.

The work presents \emph{conceptual innovations}, i.e., using recommendations as an untrusted prediction source for caching, and using different online caching algorithms in an optimistic meta-learning algorithm; as well as \emph{technical contributions} such as the tightened bound of optimistic proximal FTRL (Theorem  \ref{th:regret1}) and the new optimistic proximal FTRL algorithm with budget constraints (Theorem  \ref{th:regret-e}). While we focus on data caching, the proposed algorithms can be directly applied to caching of services and code libraries in edge computing systems.

\textbf{Paper Organization}. The rest of this work is organized as follows. Sec. \ref{sec:model} introduces the system model and states formally the problem. Sec. \ref{sec:bipartite} presents the optimistic online caching policy for the bipartite graph, and Sec. \ref{sec:elastic} presents the respective policy and results for the case of elastic caching systems. Sec. \ref{sec:exps} introduces the meta-learning framework with the inclusion of multiple predictors, and Sec. \ref{sec:evaluation} presents our numerical evaluation of the proposed algorithms using synthetic and real traces. We conclude in Sec. \ref{sec:conclusions}.
\section{System model and Problem Statement}\label{sec:model}

\begin{figure}[!t]
	\centering
	\includegraphics[width=2.80in]{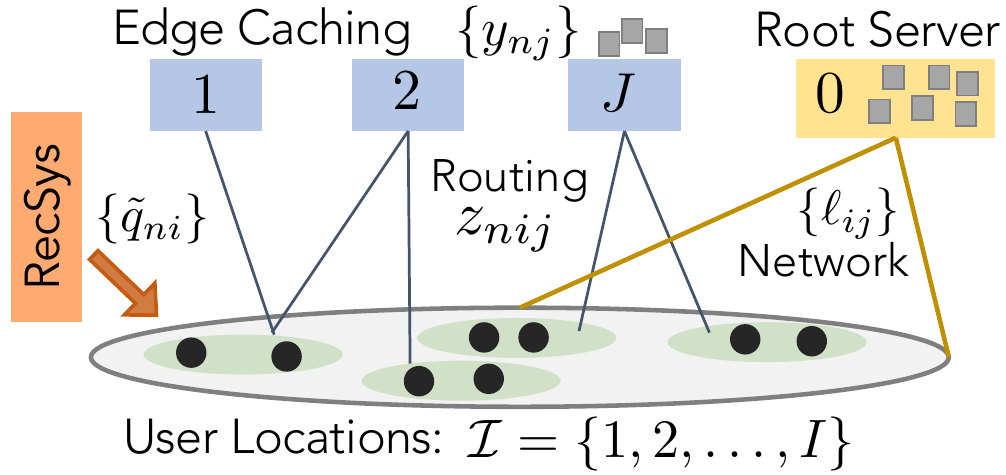} 
	\caption{\textbf{System Model}. A network of $\mathcal J$ caches serves file requests  from a set $\mathcal I$ of user locations. Unserved requests are routed to the Root Server. Caching decisions are aided via the recommendations provided by the  rec-sys.}
	\label{fig:system}
\end{figure}
\subsection{Model Preliminaries}

\textbf{Notation}. We use calligraphic capital letters, e.g., $\mathcal X$, to denote sets. Vectors are denoted with regular small letters, e.g., $a$, and we use the subscript $t$ to highlight a vector's dependence on a specific time slot e.g., $a_t$. When a compenent of the vector is indexed, the subscript is repurposed to denote that component's (multi-)index, while the $t$ moves to the superscript. i.e., for the $d$-dimensional vector $a_t$ we write $a_t=(a^t_1, a^t_2, \ldots, a^t_d)$. We denote with $\{a_t\}_{t=1}^T$ the sequence of vectors or parameters from slot $t=1$ up to slot $T$; whenever the horizon is not relevant we use $\{a_t\}_t$. We also use the shorthand sum notation $b_{1:t} = \sum_{i=1}^T b_i$. We also denote with $[T]$ the integer set ${1,2, \cdots, T}$. The notation for the model and the subsequently presented algorithms is summarized in Table \ref{table:notation}.

\begin{table}
	\caption{Key Notation}
	\centering%
	\renewcommand{\arraystretch}{1.25}
	\begin{tabularx}{0.485\textwidth}{|c|Y|}
		\hline
		\textbf{Parameters} & \textbf{Physical Meaning}\\
		\hline\hline %
		\multicolumn{2}{|c|}{\textit{\textbf{Caching Network}}}\\ \hline
		$\mathcal{J} (J)$ & Set (number) of caches \\ \hline
		$\mathcal{I} (I)$ & Set (number) of user locations \\ \hline
		$C_j$ & Capacity of cache $j$. $C_j \leq C, \forall j\!\in\! \mathcal J$ \\ \hline
		$\ell_{ij}$ & Indicator for connectivity of location $i$ to cache $j$ \\ \hline
		$N$ & Number of files \\ \hline
		$q_{ni}^t$ & Request issued by user $i$ for file $n$ at slot $t$\\
		\hline
        $w_{nij}$ & Utility for routing a unit of file $n$ from cache $j$ to $i$
        \\ \hline
        $P$ & Number of predictors \\ \hline
        $c_t$ & Gradient of the utility function 
         \\ \hline
        $\tilde{c}_t$ & A prediction for $c_t$\\ \hline
        $s_j^t$ & Price for unit storage in cache $j$\\ \hline
        \multicolumn{2}{|c|}{\textit{\textbf{ Decision Variables}}}\\ \hline
		$y_{nj}^t$ & Portion of file $n$ stored at cache $j$ at slot $t$ \\ \hline
		$z_{nij}^t$ & Portion of file $n$ routed from cache $j$ to $i$ at slot $t$ \\ \hline
		$x$ & A shorthand for the concatenated variables $(y, z)$ \\ \hline
		 \multicolumn{2}{|c|}{\textit{\textbf{ Learning Algorithms}}}\\ \hline
		$r_t(\cdot)$ & A strongly convex regularizer function (for $t\geq1$)\\ \hline
		$h_t$ & Prediction error $\|\tilde{c}_t - c_t\|$\\ \hline
		$\sigma_t$ & The change in the aggregated root of prediction error $\sqrt{h_{1:t}} - \sqrt{h_{1:t-1}}$\\ \hline
		$F_t$ & Experts performance vector at slot t\\ \hline	
		$u_t$ & Weight vector used to combine experts proposals\\
		\hline
	\end{tabularx}
	\label{table:notation} 
\end{table}

\textbf{Network}. The caching network includes a set of edge caches ${\cal J}\!=\!\{1,2,\dots, J\}$ and a root cache indexed with 0, as shown in Fig. \ref{fig:system}. The file requests emanate from a set of user locations ${\cal I}=\{1,2,\dots, I\}$. The connectivity between $\mathcal I$ and $\mathcal J$ is modeled with parameters $\ell\!=\!\big(\ell_{ij}\in \{0,1\}: i\!\in\!\mathcal{I}, j\!\in\!\mathcal{J} \big)$, where $\ell_{ij}\!=\!1$ if cache $j$ can be reached from location $i$. We consider the general case where the caches have overlapping coverage; thus, each user can be (potentially) served by one or more edge caching. The root cache is within the range of all users in $\mathcal{I}$. This is a general non-capacitated bipartite model, see \cite{paschos-book} for an overview of caching models; and extends the celebrated femtocaching model \cite{femtocaching} since the link qualities (which are captured through the utility gains; see below) not only may vary with time, but can do so in an arbitrary (i.e., non-stationary) fashion. This latter feature is particularly important for the realistic modeling of volatile wireless edge caching systems \cite{paschos-magazine, 9253571}.

\textbf{Requests}. The system operation is time slotted, $t\!\!=\!\!1,2,\dots,\!T$. Users submit requests for obtaining files from a library $\mathcal{N}$ of $N$ files with unit size; we note that the analysis can be readily extended to files with different sizes, This will be made clear in Sec. \ref{subsec:ps}. Parameter $q^{t}_{ni}\!\in\!\{0,1\}$ indicates the submission of a request for file $n \! \in \! \mathcal{N}$ by a user at location $i \! \in \!\mathcal{I}$ in the beginning of slot $t$. At each slot we assume there is one request\footnote{The proposed policies will still deliver the same regret guarantees when requests are batched before an update. However, the Lipchitz constant will be scaled according to the batch size, and this will affect accordingly the constant factor of the guarantees.}; i.e., the caching decisions are updated after every request, as in LFU and LRU policies,  \cite{giovanidis-mLRU, leonardi-implicit}. Hence, the request process comprises successive vectors $q_t\!=\!(q_{ni}^t\!\in\!\{0,1\}: n\!\in\!\mathcal N, i\!\in\!\mathcal I)$ from the set:
\begin{equation}
	\mathcal{Q}=\bigg\{q\in \{0,1\}^{N\cdot I} ~\Big |~ \sum_{n\in\mathcal{N}}\sum_{i\in\mathcal{I}} q_{ni}=1\bigg\}. \notag
\end{equation}
We make no assumptions for the request pattern; it might follow a fixed or time-varying distribution that is unknown to the system; and can be even selected strategically  by an \emph{adversary} aiming to degrade the caching operation. If a policy's performance is satisfactory under this model, it is ensured to achieve (at least) the same performance  for other request models.

\textbf{Recommendations}. There is a recommender system (\emph{rec-sys}) that suggests files to each user $i\!\in\!\mathcal I$, see \cite{netflix} for the case of Netflix. User $i$ requests one of the recommended files with a certain probability that captures the user's propensity to follow the recommendations. Unlike prior works that consider these probabilities fixed (or assume that the recommendation following event is a stationary stochastic process) \cite{akis-wowmom20, jordan-tmc}, we model them as unknown and possibly time-varying. Namely, no assumption on their quality is guaranteed to remain valid.

A key point in our approach is that the content recommendations, if properly leveraged, can serve as \emph{predictions} for the next-slot requests which are otherwise unknown. We denote with $\tilde q_t$ the prediction for the request $q_t$ that the system will receive at the beginning of slot $t$, and we assume that $\tilde q_t$ is available at the end of slot $t\!-\!1$, i.e., when the rec-sys provides its recommendations. Essentially, the recommender system is an approach for obtaining prediction for the next request. The recommendation can be mapped to predictions in different ways. For example, the caching system can set $\p{q}_{\hat n\hat i}^{\ t+1} = 1$ and $\p q_{ni}^{\ t+1} = 0, \forall (n,i)\neq (\hat n, \hat i)$, where $(\hat n, \hat i)$ is the request with the highest predicted probability (top recommended file)\footnote{Note that our caching policy is orthogonal to the mechanism that maps the recommendations to predictions. Namely, our results will be stated in terms of the prediction error.}. In section \ref{sec:exps}, we study the case where a set $\mathcal{P}=\{1, \ldots, P\}$ of $P$ different predictors (other than the recommendatino system) are available, each one offering a prediction $q_t^{(p)}, p \in \mathcal{P}$ at every slot $t$.

\textbf{Caching}. Each cache $j\!\in\!\mathcal J$ stores up to $C_j\!<<\!N$ files, while the root cache stores the entire library, i.e., $C_0\!\geq \!N$. We also define $C\!=\! \max_{j\in\mathcal J} C_j$. Following the standard femtocaching model \cite{femtocaching}, we perform caching using the \emph{Maximum Distance Separable} (MDS) codes, where files are split into a fixed number of $F$ chunks, which include redundancy chunks. A user can decode the file if it receives any $F$-sized subset of its chunks. For large values of $F$, the MDS model allows us to use continuous caching variables.\footnote{Large files are composed of thousands of chunks, leading to a small chunk size (compared to the original file). This induces practically negligible errors in the utility function \cite[Sec. 3.3]{paschos-book}. In addition, even for exact discrete caching, relaxing the integrality constraints and solving the continuous version is an essential first step which is then followed by a randomized rounding technique (see, e.g., \cite[Sec. 6]{tareq-jrnl}).} Hence, we define the  variable $y_{nj}^t\!\in\![0,1]$ which denotes the portion of $F$ chunks of file $n\!\in\!\cal N$ stored at cache $j\!\in\! \cal J$, and we introduce the $t$-slot caching vector $y_t\!=\!(y_{nj}^t: n\!\in\! \mathcal N, j\!\in\!\mathcal J)$ that belongs to set: 
\[
\mathcal{Y}=\bigg\{y\in [0,1]^{N\cdot J} ~\Big|~ \sum_{n\in\mathcal{N}}y_{nj}\leq C_j, ~j\in \mathcal J\bigg\}.
\]

\textbf{Routing}. Since each user location $i\in\mathcal{I}$ may be connected to multiple caches, we need to introduce  routing variables. Let $z_{nij}^t$ denote the portion of request $q_{ni}^{t}$ served by cache $j$. In the MDS caching model the requests can be simultaneously routed from multiple caches and, naturally, we restrict\footnote{This practical constraint is called the \emph{inelastic} model and compounds the problem, cf. \cite{abhishek-sigm20} for the simpler elastic model.}  the amount of chunks not to exceed $F$. Hence, the $t$-slot routing vector $z_t=(z_{nij}^t\!\in\![0,1]: n\!\in\!\mathcal N, i\!\in\!\mathcal I, j\!\in\!\mathcal J)$ is drawn from:
\[
\mathcal{Z}=\bigg\{z\in [0,1]^{N\cdot J\cdot I} ~\Big|~ \sum_{j\in\mathcal{J}}z_{nij}\leq 1, ~n\in \mathcal N, i\in\mathcal I\bigg\}.
\]
Requests that are not (fully) served by the edge caches $\mathcal J$ are served by the root server that provides the missing chunks. This decision needs not to be explicitly modeled as it is directly determined by the routing vector $z_t$.

\subsection{Problem Statement} 
\label{subsec:ps}
\textbf{Cache Utility \& Predictions}. We use parameters $w_{nij}\in[0, w]$ to model the system utility when delivering a chunk of file $n\!\in\! \cal N$ to location $i\!\in\! \cal I$ from cache $j\!\in\!\cal J$, instead of using the root server. This utility model can be used to capture bandwidth or delay savings, and other edge-caching gains in wired or wireless networks. The caching benefits can in general differ for each cache and user location, and may vary with time. Note that the cache-hit maximization problem is a special case of this setting \cite{paschos-jsac}. To streamline presentation we introduce vector $x_t\!= \!(y_t, z_t)\in \mathbb R^m$, with $m\!=\!NIJ\!+\!NJ$, and define the system utility in slot $t$ as:
\begin{equation}\label{eq:biput}
	f_t(x_t)= \sum_{n\in\mathcal{N}}\sum_{i\in\mathcal{I}}\sum_{j\in\mathcal{J}} w_{nij} q^t_{ni}z_{nij}^t\ ,
\end{equation}
and we denote its gradient $c_{t}\!=\!\nabla f_{t}(x_{t})$. As it will become clear, our analysis holds also for non-linear concave functions $f_t(x)$; this generalization is useful in case, e.g., we wish to enforce fairness in the dispersion of caching gains across the user locations \cite{jungho-wiopt}.

The main challenge in online caching is the following: at the end of each slot $t$ where we need to decide the cache configuration, the utility function $f_{t+1}$ is not available. Indeed, this function depends on the next-slot request $q_{t+1}$, which, by the time it gets revealed, $y_{t+1}$ would have already been decided and fixed\footnote{In our case, since the routing is directly shaped by the caching, this restriction affects also $z_{t+1}$.}, see \cite{paschos-infocom19, abhishek-sigm20, Lykouris-ML}. This is also the timing of the LRU/LFU policies \cite{giovanidis-mLRU, leonardi-implicit}. However, the recommendations provided to users can be used to form a predicted request $\p q_{t+1}$. Then, we can use $\p q_{t+1}$ to create a prediction for the next slot function $\p f_{t+1}(\cdot)$ through its gradient $\p c_{t+1}$, (recall that $f_{t+1}(\cdot)$ is a linear function that is parameterized by its gradient $c_{t+1}$).


\textbf{Benchmark}. In such learning problems, it is important to understand the objective that our algorithm aims to achieve. If we had access to an oracle for the requests $\{q_t\}_{t=1}^T$ (and the utility parameters) we could have devised the utility-maximizing static caching and routing policy $x^\star=(y^\star, z^\star)$, by solving the following convex optimization problem:
\begin{align}
\mathbb P_1:\quad \max_{x}  \,\,\,\,\,\,& \sum_{t=1}^T f_t(x) \label{eq:opt1a}\\
\text{s.t.   }&\,\, z_{nij} \leq y_{nj}\ell_{ij}, \quad i\in\mathcal I, j\in\mathcal J, n\in\mathcal N, \label{eq:opt1b}\\
	& \,\,z\in \mathcal Z, \,\,\,\, y\in\mathcal Y
	\label{eq:opt1c},
\end{align} 
where \eqref{eq:opt1b} ensure the routing decisions for each requested file use only caches that store enough chunks of that file.Note that the case of files of different sizes corresponds to replacing the set $\mathcal{Y}$ with $\mathcal{Y}' = \{y\in [0,1]^{N\cdot J} ~|~ \sum_{n\in\mathcal{N}}v_n y_{nj}\leq C_j, ~j\in \mathcal J\}$ For some general size vector $v \in \mathbb{R}_+^N$. Such a change will not affect the mathematical characteristics of the above optimization problem (both sets correspond to linear constraints).

Let us define the convex set of constraints:
\begin{align}
\mathcal X=\big\{\left\{\mathcal Y\times \mathcal Z\right\}\cap \{\eqref{eq:opt1b} \}\big\} \label{eq:const_set}
\end{align}
that we will use henceforth to streamline presentation.

Clearly, this hypothetical solution $x^\star$ can be designed only with \emph{hindsight} and is the benchmark for evaluating our online learning policy $\pi$ which outputs $\{x_t\}_t$. Thus, we use the regret metric:
\begin{align}\label{eq:regret}
R_T(\pi)= \sup_{ \{f_t\}_{t=1}^T }\left[ \sum_{t=1}^T f_t\big(x^\star\big)-\sum_{t=1}^T f_t\big(x_t\big)\right],
\end{align}

which quantifies the performance gap of $\pi$ from $x^\star$, for any possible sequence of requests or, equivalently, functions $\{f_t\}_t$. Our goal is to find a policy that achieves sublinear regret, $R_T(\pi)\!=\!o(T)$, thus ensuring the average performance gap $R_T/T$ will diminish as $T$ grows. This policy, similar to other online policies, decides $x_{t+1}$ at the end of each slot $t$ using the previous utility functions $\{f_\tau\}_{\tau=1}^t$ and the next-slot prediction $\tilde f_{t+1}$ devised from the rec-sys.

Lastly, note that, in principle, the regret metric can be negative. This is especially true for optimistic policies. To see why, recall that $x^\star$ is the best \emph{fixed} caching configuration, whereas $x_t$ is allowed to change for each $t$. Hence, it might happen, e.g., that $x_t$ performs better than $x^\star$ on some steps, $f_t(x_t) \geq f_t(x^\star)$, while performing similar to $x^\star$ in the remaining ones $f_t(x_t) \approx f_t(x^\star)$. Of course, the occurrence of such an event depends on the request sequence and the policy that determines $x_t$. In the following sections, we show that optimism can greatly decrease the upper bound on $R_T$, increasing the chances of negative regret.

\section{Optimistic Bipartite Caching}\label{sec:bipartite}

Unlike recent caching solutions that rely on Online Gradient Descent (OGD)  \cite{paschos-infocom19} or on the Follow-the-Perturbed-Leader (FTPL) policy \cite{abhishek-sigm20}, our approach draws from the \emph{Follow-The-Regularized-Leader} (FTRL) policy, cf. \cite{mcmahan-survey17}, appended with prediction-adaptive regularizers. A key element in our proposal is  the \emph{optimism} emanating from the availability of predictions, namely the content recommendations that are offered to users by the rec-sys in each slot.

Let us begin by defining the proximal regularizers\footnote{A proximal regularizer induces a proximal mapping for the objective function; see \cite[Ch. 6.1]{beck-book} for the formal definition.}: 
\begin{align}
r_0(x)=\bm I_{\mathcal X}(x), \quad	r_t(x)=\frac{\sigma_t}{2}\|x-x_t\|^2,\,\,t\geq 1 \label{regula}
\end{align}
where $\|\cdot \|$ is the Euclidean norm, and $\bm I_{\mathcal X}(x)\!=\!0$ if $x\!\in\! \mathcal X$ and $\infty$ otherwise. We apply properly selected regularizing parameters, which change the strong convexity of $r_t$ according to the predictions quality until $t$, and also ensure $r_t(x)\geq 0, \forall t$, namely:
\begin{align}
	&\sigma_1 = \sigma \sqrt{h_1}, \quad \sigma_t=\sigma \left( \sqrt{ h_{1:t} } - \sqrt{h_{1:t-1} }\right), \quad t\geq2  \label{regulb}
	\\ 
	&\text{with} \quad h_t=\|c_t-\p c_t\|^2, \notag 
\end{align}
where $\sigma\!\geq \!0$, $c_t\!=\! \grad f_t(x_t)$, and we used the shorthand sum notation $h_{1:t}\!=\!\sum_{i=1}^th_i$ for the aggregate prediction errors during the first $t$ slots.  The basic step of the algorithm is: 
\begin{align}
x_{t+1}=\arg\min_{x\in\mathbb R^m}\Big\{ r_{0:t}(x)- (c_{1:t} + \p c_{t+1})^\top x	\Big\},   \label{proxy-step}
\end{align}
which calculates the decision vector using past utility observations $c_{1:t}$, the aggregate regularizer $r_{0:t}(x)$ and the prediction $\p c_{t+1}$. The update employs the negative gradients as it concerns a maximization problem. Henceforth, we refer to \eqref{proxy-step} as the \emph{optimistic} FTRL (OFTRL) update. 

Policy $\pi_{obc}$ is outlined in Algorithm \ref{alg1}. In each iteration, OBC solves a convex optimization problem, \eqref{proxy-step}, involving a projection on the feasible set $\mathcal X$ (via $r_0(x)$). For the latter, one can rely on fast-projection algorithms specialized for caching, e.g., see \cite{paschos-infocom19}; while it is possible to obtain a closed-form solution for the OFTRL update for linear functions. We quantify next the performance of Algorithm \ref{alg1}. 

\begin{algorithm}[t]
	\nl \textbf{Input}: $\{\ell_{ij}\}_{(i,j)}$; $\{C_j\}_j$; $\mathcal{N}$; $x_1\!\in\!\mathcal X$; $\sigma=\sqrt{2}/D_{\mathcal X}$.\\%
	\nl \textbf{Output}: $x_t=(y_t, z_t)$, $\forall t$.\\%
	\nl \For{ $t=1,2,\ldots$  }{
		\nl Route request $q_{t}$  according to configuration $x_t$\\
		\nl Observe system utility $f_t(x_t)$ \\
		\nl Observe the new prediction $\p c_{t+1}$\\
		\nl Update the regularizer $r_{0:t}(x)$ using \eqref{regula}-\eqref{regulb}\\
		\nl Calculate the new  policy $x_{t+1}$ using \eqref{proxy-step}\\
	}
	\caption{{Optimistic Bipartite Caching ($\pi_{obc}$)}}\label{alg1}
	\setlength{\intextsep}{0pt} 
\end{algorithm} 

\begin{mdframed}
\begin{theorem}\label{th:regret1}
Algorithm 1 ensures the regret bound:
\begin{align}
R_T\leq 2\sqrt{1\!+JC} \sqrt{\sum_{t=1}^T  \|c_t-\p c_t\|^2 }.\notag
\end{align}
\end{theorem}
\end{mdframed}

For the proof, we modify the ``strong FTRL lemma" \cite[Lemma 5]{mcmahan-survey17} by adding predictions for next-slot utility function. The following lemma bounds the regret in terms of the difference between two values of a strongly convex function evaluated at $x_t$ and at $x_{t+1}$, for each $t$.

\begin{lemma}{(Optimistic Strong FTRL Lemma)} \label{lemm:OSFTRL} Let $v_t(x) = -c_{t}^\top x_t + r_t(x_t)$, and $v_{0:t}(x) = -c_{1:t}^\top x_t + r_{0:t}(x_t)$. Let $x_{t+1}$ be selected according to \eqref{proxy-step}. Then:
\begin{multline}
    R_T \leq r_{0:T}(x^\star) + \sumT v_{0:t}(x_t) - v_{0:t}(x_{t+1}) -r_t(x_t)
    \\
    + \tilde c_{T+1}^\top (x_{T+1} - x^\star) 
    \label{eq:OSFTRL}
\end{multline}
\end{lemma}

\begin{proof}[Proof of Lemma \ref{lemm:OSFTRL}]
\begin{align}
    &\sumT v_t(x_t) - (v_{0:T}(x^\star) - \tilde c_{T+1}^\top x^\star) \notag
    \\
    &=\sumT \left( v_{0:t}(x_t) - v_{0:t-1}(x_t)  \right) \notag -(v_{0:T}(x^\star) - \tilde c_{T+1}^\top x^\star) 
    \\
    &\leq \sumT v_{0:t}(x_t) + \sumT \left( - v_{0:t-1}(x_t)\right)  \notag
    \\
    &-(v_{0:T}(x_{T+1}) - \tilde c_{T+1}^\top x_{T+1}) \notag \quad\quad (\text {by def. of $x_{T+1}$}) 
    \\
    &\stackrel{\text{(a)}}{\leq}\sumT v_{0:t}(x_t)  - v_0(x_1)+ \sum_{t=1}^{T-1} \left( - v_{0:t}(x_{t+1}) \right) \notag
    \\
    &- v_{0:T}(x_{T+1})+ \tilde c_{T+1}^\top x_{T+1} \notag 
    \\
    &\stackrel{\text{(b)}}{\leq}\sumT \big(v_{0:t}(x_t)  - v_{0:t}(x_{t+1})\big) + \tilde c_{T+1}^\top x_{T+1} 
\end{align}
where equality (a) follows by reindexing the second sum (i.e., changing the sum index from $t$ to $t+1$), and inequality (b) by dropping the non-positive term $-v_0(x_1) = - r_0(x_1)$ and appending the term $- v_{0:T}(x_{T+1})$ to the second sum.
Thus, we have that:
\begin{align}
    &\sumT v_t(x_t) - v_{0:T}(x^\star) + \tilde c_{T+1}^\top x^\star \notag
    \\
    &\leq\sumT \big(v_{0:t}(x_t)  - v_{0:t}(x_{t+1})\big) + \tilde c_{T+1}^\top x_{T+1}.
\end{align}
Expanding the definition of $v_t(x_t)$ and rearranging give the regret inequality.
\begin{multline}
\renewcommand{\qedsymbol}{}
c_{0:T}^\top x^\star - \sumT c_t^\top x_t \leq r_{0:T}(x^\star) + \sumT \big(v_{0:t}(x_t)  - v_{0:t}(x_{t+1})\big)
\\
-r_t(x_t) + \tilde c_{T+1}^\top (x_{T+1} - x^\star)\notag.
\end{multline}
\end{proof}

With the weak assumption that the caching system will be notified upon the serving of the last request, we can set $\tilde c_{T+1}=0$ and hence cancel the last term in the above inequality. Otherwise, it will be an additional constant factor in the regret bound\footnote{It is possible to slightly change the semantics of the algorithm to avoid this rare case by making the adversary first commit and hide the cost function, and then the learner picks the action, see discussion of \cite[Thm 7.29] {orabona2021modern}.}.
Next, we will make use of the following results to bound each $v_{0:t}(x_t) - v_{0:t}(x_{t+1})$ term:
\begin{lemma}{\cite[Lemma 7]{mcmahan-survey17}} \label{lemm:OSFTRL2}
let $\phi_1\!:\! \mathbb{R}^n \to\! \mathbb{R}$ be a convex function such that $x_1 \!=\! \arg\min_{x} \phi_1$. Let $\psi$ be a convex function such that  $\phi_2(x) \!=\! \phi_1(x) \!+\! \psi(x)$ is strongly convex w.r.t norm $\|\cdot\|$. Then, for any $b \in \partial \psi(x_1)$ and $x'$, we have that $\phi_2(x_1) \!-\! \phi_2(x^{'})\! \leq \!\frac{1}{2}\|b\|_\star^2$.
\end{lemma}
Now we are ready to prove Theorem 1:
\begin{proof}[Proof of Theorem 1]
We start by applying Lemma \ref{lemm:OSFTRL2} to the result in \eqref{eq:OSFTRL}. Namely, we select:
\begin{align*}
&\phi_1(x)=v_{0:t}(x) + {c_t^\top}{x_t} -{\tilde c_t^\top}{x_t},\,\, \text{and}\\
&\phi_2(x)=\phi_1(x) - {c_t^\top}{x_t} +{\tilde c_t^\top}{x_t}. 
\end{align*}
This way, we have that $x_t\! =\! \arg\min{\phi_1(x_t)}$,\  $\phi_2(x)\! =\! v_{0:t}(x)$,\  $\psi(x)\! =\! (-c_t+\tilde c_t)^\top x$, \ and $(-c_t+\tilde c_t)\! \in\! \partial \psi(x)$. 

Then, dropping the non-positive terms $-r_t(\cdot)$ in \eqref{eq:OSFTRL}, setting $\tilde{c}_{T+1}=0$, and defining the norm
$\|\cdot\|_{(t)} = \sqrt{\sigma_{1:t}}\|x\| $ so that the regularizer $r_{1:t}(x)$ is 1-strongly-convex w.r.t. $\|\cdot\|_{(t)}$ we get:
\begin{align}
R_T\leq r_{1:T}(x^\star)+\frac{1}{2}\sum_{t=1}^T\|c_t-\p c_t \|_{(t),\star}^2,\,\,\,\,\forall x^\star \in \cal X. \label{mohri-rt}
\end{align}
Now, we have that $r_t\leq \frac{\sigma_t}{2}D_{\mathcal X}^2$, where $D_{\mathcal X}$ is the Euclidean diameter $D_{\mathcal X}$, i.e., $\forall x, x_t\! \in\! \cal X$:
\begin{align*}
\|x-x_t\|^2&=\sum_{n,j}(y_{nj}-y_{nj}^t)^2+\sum_{n,i,j}(z_{nij}-z_{nij}^t)^2 
\\
&\stackrel{(a)}{\leq} \sum_{n,j}|y_{nj}-y_{nj}^t|+\sum_{n,i,j}|z_{nij}-z_{nij}^t| \notag
\\
&\stackrel{(b)}{\leq}\notag 2(JC+1)\triangleq D_{\mathcal X}^2 \notag
\end{align*}
where $(a)$ holds as $y_{nj}, z_{nij}\!\in\![0,1], \forall n,i,j$; $(b)$ holds by the triangle inequality and definitions of $\mathcal Y$, $C\!\triangleq\! \max_{j} C_j$, and the fact that the routing variables differ at one coordinate only. To see why, recall that we serve one request per time slot and we set the routing variable $z_t$ after observing that request $q_t$. Hence, we can modify the routing variables and set $z^\star_{nij} =  z^t_{nij} = 0\  \forall n \neq n', i \neq i', j \neq j'$, where $q_{n'i'}=1 \wedge \ell_{i'j'=1}$. Due to the structure of $f_t(\cdot)$, the utility of these modified $z^\star_{nij}$ and $z^t_{nij}$ will not change (compared to their utility before modification). In words, knowing the requested file $n'$, and the location from which it is requested $i'$, there will be no utility from routing a non-requested file $n\neq n'$, or routing from a non-connected cache $j\neq j'$. Thus, we can zero these variables and get a smaller value for $D_{\mathcal{X}}$, without affecting the utility values.

Using this diameter bound, and the fact that the dual norm of $\|x\|_{(t)}$ is $\|x\|_{(t),\star} = \|x\|/\sqrt {\sigma_{1:t}}$, inequality \eqref{mohri-rt} can be written as:  
\begin{align}
R_T &\leq \frac{\sigma_{1:T}}{2} D_{\mathcal X}^2
+ \frac{1}{2} \sum_{t=1}^T \frac{h_t}{\sigma_{1:t}} \label{eq:reg_sigm}.
\end{align}
Note that the sum $\sigma_{1:t}$ telescopes and evaluates to $\sigma\sqrt{h_{1:t}}$. Using this observation and substituting it in \eqref{eq:reg_sigm}, and combining it with \cite[Lem. 3.5]{auer-2002} to bound the second term as follows $\sum_t^T h_t/\sqrt{h_{1:t}}\!\leq\! 2\sqrt{h_{1:T}}$, we eventually get:
\begin{align}
R_T&\leq \frac{\sigma}{2}  \sqrt{h_{1:T}} D_{\mathcal X}^2
+  \frac{1}{\sigma} \sqrt{h_{1:T}}\stackrel{(a)}\leq \sqrt{2} D_{\mathcal{X}} \sqrt{ h_{1:T}}, \notag
\end{align}
where $(a)$ is obtained by setting $\sigma=\sqrt{2}/D_{\mathcal X}$. Finally, substituting the actual diameter value, namely $D_{\mathcal X}\!\! =\!\! \sqrt{2(JC+1)}$, completes the proof.
\end{proof}

\textbf{Discussion}. Theorem \eqref{th:regret1} shows that the regret does not depend on the library size $N$ and is also modulated by the quality of the predictions; accurate predictions tighten the bound, and in the case of perfect predictions, i.e., when users follow the recommendations, we get negative regret $R_T \! \leq \! 0, \forall T$, which is much stronger than the sub-linear growth rates in other works \cite{paschos-infocom19,8999784}. In fact, even when predictions fail for constant number of time slots $L \in [T]$, the regret will be constant of the same order $R_T \leq O(L)$, which is still a significant improvement over any time-dependent bound.

On the other hand, for worst-case prediction, we can still use the bound $\|c_t-\p c_t \|^2 \leq 2w^2$, and get:
\[
R_T\leq 2\sqrt{2}w\sqrt{JC+1}\sqrt{T}=O(\sqrt{T})
\]
i.e., the regret is at most a constant factor worse than the regret of those policies that do not incorporate predictions\footnote{The factor is $\sqrt{2}$ compared to the ``any-time" version of the bound that do not use predictions, and $2$ compared to those that assume a known $T$.}. Thus, OBC offers an efficient and safe approach for incorporating predictions in cases where we are uncertain about their accuracy, e.g., either due to the quality of the rec-sys or the behavior of users.

Another key point is that the \emph{utility parameters might vary with time} as well. Indeed, replacing $w_t=(w_{nij}^t\!\leq \! w, n\!\!\in\!\!\mathcal N, i\!\!\in\!\!\mathcal I, j\!\!\in\!\!\mathcal J)$ in $f_t(x_t)$ does not affect the analysis nor the bound. This is important when the caching system employs a wireless network where the link capacities vary, or when the caching utility changes; and we note that this utility can be even \emph{file-specific}. Parameters $w_t$ can be also unknown when $x_t$ is decided, exactly as it is with $q_t$, and they can be predicted using e.g., channel measurements. Essentially the proposed model drops several restricted assumptions of prior works regarding, not only the knowledge of request rates/densities, but also about the system state, link quality, and user utilities. Finally, we observe that in case the algorithm is used to optimize the operation of an edge computing system, these parameters can capture the potentially time-varying utility of each computation, for each service ($n$) and each pair of user - cache. 

On a technical note, Lemma \ref{lemm:OSFTRL} presents a novel theoretical result and enables the improvement of the best known proximal OFTRL bound of \cite{mohri-aistats2016} by a constant factor of $\sqrt{2}$. This also opens the door for leveraging the modular analysis tools developed in \cite{mcmahan-survey17} in the optimistic OCO framework for different special cases of the utility functions. This technical result is of independent interest.

Lastly, in the case of more than one request per time slot (e.g., $B$ requests), the Euclidean norm would instead be $D^2_{\mathcal{X}} = 2(JC+B)$. Therefore, the same results hold (in terms of the regret being always sub-linear and commensurate with the prediction accuracy). However, the worst case bounds will of course be scaled by a factor of $\sqrt{B}$ since now we would use $\|c_t - \tilde c_t \|^2  \leq 2 B\ w^2$.

\section{Optimistic Caching in Elastic Networks}\label{sec:elastic}

We extend our analysis to \emph{elastic} caching networks where the caches can be resized dynamically. Such architectures are important for two reasons. Firstly, there is a growing number of small-size content providers that implement their services by leasing storage on demand from infrastructure providers \cite{amazon-cache}; and secondly, CDNs often resize their caches responding to the time-varying user needs and operating expenditures \cite{elastic-cdn}.

\begin{algorithm}[t]
	\nl \textbf{Input}: $\{\ell_{ij}\}_{(i,j)}$, $\{C_j\}_j$, $\mathcal{N}$, $\lambda_1\!=\!0$, $x_1\!\in\!\mathcal X_e$.\\%
	\nl \textbf{Output}: $x_t=(y_t, z_t)$, $\forall t$.\\%
	\nl \For{ $t=1,2,\ldots$  }{
		\nl Route  request $q_{t}$  according to configuration $x_t$\\
		\nl Observe system utility $f_t(x_t)$ and cost $g_t(x_t)$ \\
		\nl Update the budget parameter $\lambda_{t+1}$ using \eqref{dual-update}\\	
		\nl Update the regularizers $r_{0:t}(x)$ using \eqref{regula}-\eqref{regulb}, and $a_t\!=\!at^{-\beta}$\\
		\nl Observe prediction $\p c_{t+1}$ and price $s_{t+1}$ \\		
		\nl Calculate the new  policy $x_{t+1}$ using \eqref{primal-update} \\
	}
	\caption{{Optimistic Elastic Caching ($\pi_{oec}$)}}\label{alg2}
\end{algorithm}

We introduce the $t$-slot price vector $s_t\!=\!(s_j^t\!\leq \!s, j\!\in\!\mathcal J)$, where $s_j^t$ is the leasing price per unit of storage at cache $j$ in slot $t$, and $s$ its maximum value. In the general case, these prices may change arbitrarily over time, e.g., because the provider has a dynamic pricing scheme or the electricity cost changes \cite{giannakis-elasticJSAC19, jungho-wiopt}; hence the caching system has access only to $s_t$ at each slot $t$. We denote with $B_T$ the budget the system intends to spend during a period of $T$ slots for leasing cache capacity. The objective is to maximize the caching gains while satisfying:
\begin{align}
\sum_{t=1}^T g_t(x_t)=\sum_{t=1}^T \sum_{j\in\mathcal J}\sum_{n\in\mathcal N} s_j^ty_{nj}^t-B_T\leq 0.
\end{align}
In particular, the new benchmark problem in this case is:
\begin{align}
\mathbb P_2:\quad \max_{x\in\mathcal X} \sum_{t=1}^T f_t(x), \,\,\,\,\,\,\text{s.t.} \sum_{t=1}^T g_t(x)\leq 0,
\end{align}
which differs  from $\mathbb P_1$ due to the  leasing constraint. 

Indeed, in this case the regret is defined as:
\begin{align}\label{eq:regret-elastic}
	R_T^{(e)}(\pi)= \sup_{ \{f_t\}_{t=1}^T }\left[ \sum_{t=1}^T f_t\big(x^\star\big)-\sum_{t=1}^T f_t\big(x_t\big)\right],
\end{align}
where 
\[
x^\star\! \in\! \mathcal X_e\triangleq\{ x \in \mathcal X \mid \eqref{eq:opt1b}, g_t(x)\leq 0, \forall t \},
\]
i.e., $x^\star$ is a feasible point of $\mathbb P_2$ with the additional restriction to satisfy $g_t(x)\!\leq\!0$ in every slot. In the definition of $\mathcal{X}$, $C$ now denotes the maximum leasable space. Learning problems with time-varying constraints are hard, see impossibility result in \cite{tsitsiklis}, and hence require such additional restrictions on the selected benchmarks. We refer the reader to \cite{giannakis-TSP17} for a related discussion, and to \cite{paschos-icml, victor} for different benchmarks. Finally, apart from $R_T^{(e)}$, we need also to ensure sublinear growth rate for the budget violation:

\begin{equation}
V_T^{(e)}=\sum_{t=1}^T \left[ g_t(x_t) \right]_+. \notag
\end{equation}
To tackle this new problem we follow a saddle point analysis, which is new in the context of OFTRL. 

Namely, we first define a Lagrangian-type function by relaxing the budget constraint and introducing the dual variable $\lambda\geq 0$:

\begin{align}
\mathcal L_t(x,\lambda)\!=\!\frac{\sigma_t}{2}\|x\!-x_t\|^2 \!- f_t(x) \! +\lambda g_t(x)\!- \frac{\lambda^2}{a_t}. \label{eq:lag}
\end{align}
The last term is a non-proximal regularizer for the dual variable; and we use $a_t\!=\!at^{-\beta}$, where parameter $\beta\!\in\! [0,1)$ can be used to prioritize either $R_T^{(e)}$ or $V_T^{(e)}$. The main ingredients of policy $\pi_{oec}$ are the saddle-point iterations:
\begin{align}
\lambda_{t+1}=\arg\max_{\lambda\geq 0}\left\{-\frac{ \lambda^2}{a_{t+1}} + \lambda\sum_{i=1}^tg_i(x_i) \right\},\label{dual-update}	
\end{align}
\begin{align}
\!\!x_{t+1}\!=\!\arg\min_{x\in\mathbb R^m}\Bigg\{\!r_{0:t}(x)\!+\!\big( \sum_{i=1}^{t+1}\!\lambda_is_i -c_{1:t}\!-\p c_{t+1} \big)^\top\! x \Bigg\},
\label{primal-update} 
\end{align}
and its implementation is outlined in Algorithm \ref{alg2}. Note that we use the same regularizer as in Sec. \ref{sec:bipartite} for the primal variables $x_t$, while $\lambda_t$ modulates the caching decisions by serving as a \emph{shadow price} for the average budget expenditure. 

The performance of Algorithm OEC is characterized in the next theorem.
 
\begin{mdframed}
\begin{theorem} \label{th:regret-e}
Algorithm 2 ensures the bounds:
\begin{align}
& R_T^{(e)} \leq \sqrt{2}D_{\mathcal X}\sqrt{ \!\sum_{t=1}^T \|c_t\!-\p c_t\|^2 } \!+\frac{aM}{2}T^{1-\beta}	 \notag, \\
& V_T^{(e)} \leq \sqrt{\!\frac{2\sqrt{2} D_{\mathcal X}T^\beta }{a}\sqrt{\!  \sum_{t=1}^T \|c_t\!-\p c_t\|^2  } +MT \!-\! \frac{2R_T^{(e)}T^\beta}{a}} \notag,
\end{align}		
\end{theorem}
\end{mdframed}
where we have used the grouped constants $M=\frac{(sJC)^2}{1-\beta}$.
\begin{proof}
Observe that the update in \eqref{primal-update} is similar to \eqref{proxy-step} but applied to the Lagrangian in \eqref{eq:lag} instead of just the utility, and the known prices when $x_{t+1}$ is decided represent perfect prediction for $g_t(x)$. Using Theorem \ref{th:regret1} with $c_t \! - \! \ll_ts_t$ instead of $c_t$, and  $\p c_t \! - \!\ll_ts_t$ instead of $\p c_t$, we can write:
\begin{align}
\sum_{t=1}^T\! \Big(f_t(x^\star) - f_t(x_t) +\lambda_tg_t(x_t)-\lambda_tg_t(x^\star)\Big)\!\leq\! \sqrt{2}D_{\mathcal X} \sqrt{ h_{1:T}}, \notag
\end{align}
and rearrange to obtain:
\begin{align}
	R_T^{(e)}\leq \sqrt{2}D_{\mathcal X} \sqrt{ h_{1:T}} + \sum_{t=1}^T \ll_t g_t(x^\star) - \sum_{t=1}^T \ll_t g_t(x_t). \label{r1}
\end{align}
For the dual update \eqref{dual-update}, we can use the non-proximal-FTRL bound \cite[Theorem 1]{mcmahan-survey17} to write:
\begin{align}
\!\!\!\!-\!\sum_{t=1}^T \ll_t g_t(x_t) \!+\! \ll \sum_{t=1}^T g_t(x_t) \!\leq\! \frac{\ll^2}{a_T} \!+\! \frac{1}{2}\sum_{t=1}^T a_tg_t^2(x_t) \label{l1}.
\end{align}
Since $g_t(x^\star)\!\leq\! 0, \forall t$ and combining \eqref{r1}, \eqref{l1} we get:
\begin{align}
\!\!\!\!R_T^{(e)}\!\leq\!\sqrt{2}D_{\mathcal X} \sqrt{ h_{1:T} } \!- \ll\! \sum_{t=1}^T g_t(x_t) \!+\! \frac{\ll^2}{a_T} \!+\! \frac{1}{2}\sum_{t=1}^T \!a_tg_t^2(x_t) \label{eqRT1}
\end{align}
Setting $\ll\!=\!0$, using  the identity:
\[
\sum_{t=1}^Tat^{-\beta}\leq \frac{aT^{1-\beta}}{1\!-\!\beta}
\] 
and the bound $g_t(x_t)\!\leq\! sJC$, we arrive at the $R_T^{(e)}$ bound. 

For the violations, we use the following property in  \eqref{eqRT1}:
\begin{align}
	\frac{a_T}{2} \left[ \sum_{t=1}^T g_t(x_t) \right]^2_+ = \sup_{\lambda\geq 0} \left[\sum_{t=1}^T g_t(x_t)\lambda - \frac{\lambda^2}{2a_T} \right], \notag
\end{align}
Rearranging, we get:
\begin{align}
	&\frac{a_T}{2}(V_T^{(e)})^2 \leq \sqrt{2} D_{\mathcal X} \sqrt{ h_{1:T} } +\frac{a(sJC)^2}{2-2\beta}T^{1-\beta} - R_T^{(e)}. \notag
\end{align}
Finally, taking the square root yields the  $V_T^{(e)}$ bound.
\end{proof}

\textbf{Discussion}. The worst-case bounds in Theorem \ref{th:regret-e} arise when the predictions are failing. In that case, we have $\|c_t-\p c_t\|^2\leq 2w^2$ and use the bound $-R_T^{(e)}=O(T)$ for the last term of $V_T^{(e)}$, to obtain $R_T^{(e)}=O(T^\kappa)$, with $\kappa=\max\{1/2, 1-\beta\}$ while $V_T^{(e)}=O(T^\phi)$, with $\phi=\frac{1+\beta}{2}$. Hence, for $\beta=1/2$ we achieve the desired sublinear rates $R_T^{(e)}=O(\sqrt{T}), V_T^{(e)}=O(T^{3/4})$. However, when the rec-sys manages to predict accurately the user preferences, the performance of $\pi_{oec}$ improves substantially as the first terms in each bound are eliminated. Thus, for bounded $T$, we practically halve the regret and violation bounds.  

It is also interesting to observe the tension between  $V_T^{(e)}$ and $R_T^{(e)}$, which is evident from the $V_T^{(e)}$ bound and the condition $-R_T^{(e)}= O(T)$. The latter refers to the upper bound of the \emph{negative} regret, thus when it is consistently satisfied (i.e., for all $T$), we obtain an even better result: $\pi_{oec}$ \emph{outperforms} the benchmark. Another likely case is when $-R_T^{(e)}=O(\sqrt{T})$, i.e., the policy does not outperform the benchmark at a rate larger than $\sqrt{T}$. Then, Theorem \ref{th:regret-e} yields $R_T^{(e)}=O(T^\kappa)$ with $\kappa=\max\{1/2, 1-\beta\}$ while $V_T^{(e)}=O(T^\phi)$ with $\phi=\max\{1/2, 1/4+\beta/2\}$. Hence, for $\beta=1/2$ the rates are reduced to $R_T^{(e)}=O(\sqrt{T}), V_T^{(e)}=O(\sqrt{T})$.

Finally, it is worth observing that $\pi_{oec}$ can be readily extended to handle additional budget constraints such as time-average routing costs or average delays. And one can also generalize the approach to consider a budget-replenishment process where in each slot $t$ the budget increases by an amount of $b_t$ units. This is made possible due to the generality of the conditions (model perturbations can be non-stationary and correlated) under which the regret and violation bounds hold.

\section{Caching with multiple predictors} \label{sec:exps}

In this section, we consider the case where we have additional predictors, apart from the rec-sys, predicting the next-slot utility. Thus, we have a set of predictions $\{\tilde{c}_t^{(p)}, p\in\mathcal{P}\}$ at each slot $t$. To handle this setup and benefit from this abundance of predictions, we take a different approach and instead of using prediction-adaptive regularizers, as in the previous sections, we model the predictions as experts using the classical paradigm of learning through experts cf. \cite{hazan-book}. Based on this approach, we design a novel tailored optimistic \emph{meta-learning} policy to accrue the best possible caching gains.

In particular, we associate an expert to each predictor, and we will abuse notation denoting them both with $p\in \mathcal P$. We refer to these predictors-linked experts as the \emph{optimistic experts}. Every optimistic expert $p$ proposes its caching action $\{y^{(p)}_t\}_t$ at each slot $t$, by solving the following problem\footnote{To streamline the presentation, our analysis focuses on one cache, hence using only $y_t$ decisions. However, this method can be readily extended to caching networks as discussed later.} :
\begin{align}
    y^{(p)}_{t} = \argmax_{y \in \mathcal{Y}}\ \ {\tilde{c}_{t}^{(p)\top}} \ y. \label{eq:opt_action}
\end{align}
Note that \eqref{eq:opt_action} is indeed a certainty-equivalent\footnote{A certainty-equivalent program is one that considers a predicted utility vector as true and optimizes the decisions accordingly.} linear program. We denote with $R_{T}^{(p)}$ the regret of each expert w.r.t the optimal-in-hindsight caching configuration for the entire time period of $T$ slots, i.e.: 
\[
y^\star=\argmax_{y\in\mathcal{Y}} c_{1:T}^\top y.
\]

Besides the optimistic experts, we consider an expert that \emph{does not use predictions}. This special expert proposes an FTRL-based caching policy, and we refer to it as the \emph{pessimistic expert} and associate it with the special index $p=0$. The pessimistic expert proposes caching actions $\{y^{(0)}_t\}_t$ according to eq. \eqref{proxy-step}, but setting $\p c_t = 0$ for the regularization parameter $\sigma_t$ in \eqref{regulb}. Its regret is denoted with $R_{T}^{(0)}$. Our full experts set is the union of the group of optimistic experts with the pessimistic expert $\mathcal{P}^+=\{0\cup \mathcal{P}\}$. 

We aim to learn a caching policy whose regret is upper-bounded by the regret of the best expert and does not exceed the $O(\sqrt T)$ regret of the pessimistic expert. Such a methodology of modelling policies as experts has been studied in the past \cite{daron_finding}.
However, this is the first work that implements optimistic learning through an experts model. In addition, here we also make the next step and propose to include a prediction for \emph{the performance of each predictor}.

\begin{figure*}[t]
	\centering
	\includegraphics[width=0.99\textwidth]{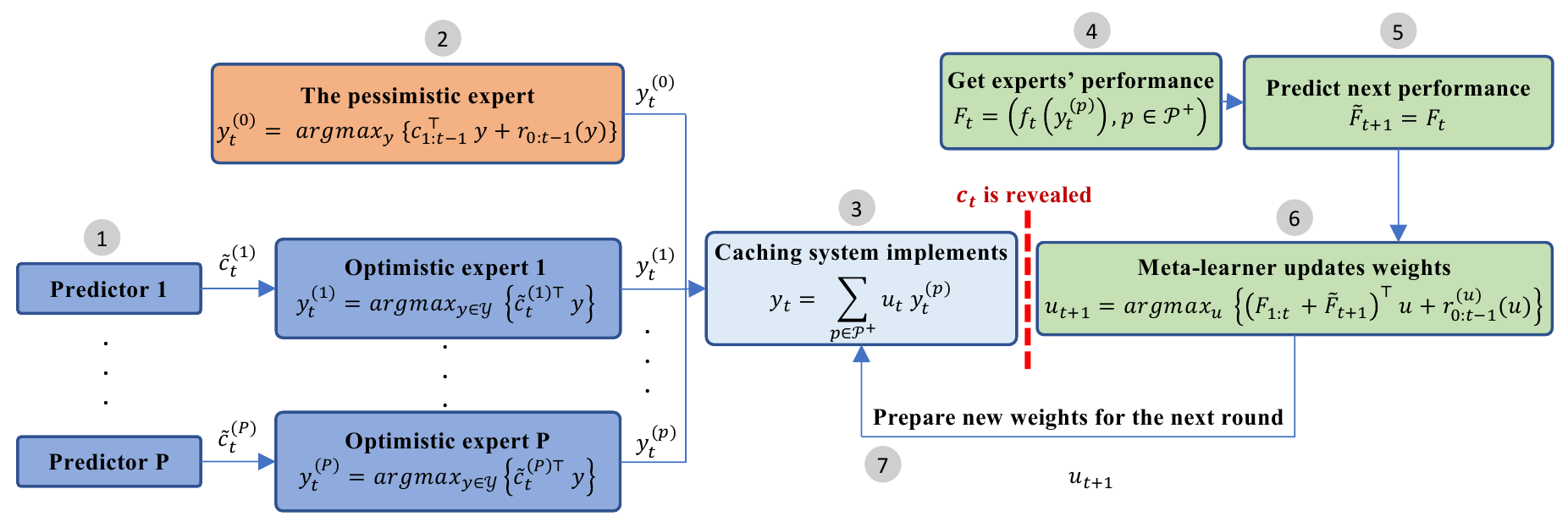}
	\caption{A decision step for the meta-learning policy $\pi_{xc}$. The policy is a combination of experts' proposals (predictors) according to their priority weights that were learned based on observations that are collected until slot $t$.}
	\label{fig:xc}
\end{figure*}

\begin{algorithm}[t]
	\nl \textbf{Input}: $C$; $y_1\!\in\!\mathcal Y$; $\sigma=\sqrt{2}/D_{\mathcal Y}$.\\%
	\nl \textbf{Output}: $y_t$, $\forall t$.\\%
	\nl \For{ $t=1,2,\ldots$  }{
		\nl Observe utility predictions $\{\tilde{c}_{t}^{(p)}, p\in\mathcal{P}\}$\\
		\nl Calculate optimistic proposals $\{\!y^{(p)}_{t}\!\}_{p\in\mathcal{P}}$ with \eqref{eq:opt_action}\\
		\nl Update $r_{0:t-1}(x)$ using \eqref{regula}-\eqref{regulb} with $\p c_{t}\! = \! 0$\\
		\nl Calculate pessimistic proposal $y^{(0)}_{t}$ with \eqref{proxy-step}\\
		\nl Serve request $q_{t}$ with meta-policy $y_t$ from \eqref{eq:comp_action}\\
		\nl Observe the utilities of experts' proposals $F_t = \left(f_t(y_t^{(p)}), p \in \mathcal{P}^+\right)$\\
		\nl Set experts' performance prediction $\tilde F_{t+1}\!=\!F_t$\\
		\nl Calculate the new weights $u_{t+1}$ using \eqref{weights_action}\\
	}
	\caption{{Experts Caching ($\pi_{xc}$)}}\label{alg3}
\end{algorithm} 

Namely, the caching decision is the convex combination of experts' proposals according to the \emph{weights} $u_t\! =\! (u^{(p)}_{t}, p \in \mathcal{P}^+)$ selected from the simplex: 
\[
\Delta_{\mathcal{P}^+} =\bigg\{u \in [0, 1]^{P+1} \bigg \vert \sum_{p \in \mathcal{P}^+} = 1\bigg\},
\]
namely:
\begin{align}
    y_{t} = \sum_{p \in \mathcal{P}^+} u^{(p)}_{t} \ y^{(p)}_{t} \label{eq:comp_action}.
\end{align}
Thus, $y_{t+1}$ remains a feasible caching vector, despite being produced by mixing all the experts' proposals. The mixing weights, $\{u_t\}_t$, are updated through a new OFTRL step, which is similar to the updates of $\pi_{obc}$ but we use the superscript $(u)$ for the involved parameters to make clear the distinction. In particular, the update is:
\begin{align}
    u_{t+1} = \argmin_{u \in \Delta_{\mathcal{P}^+}} \Big\{r_{0:t}^{(u)}(u) - (F_{1:t} + \tilde F_{t+1})^\top u   \Big\} \label{weights_action},
\end{align}
where $F_t=\big(f_t(y_t^{(p)}), p \in \mathcal{P}^+\big)$ is the $t$-slot performance vector for the experts. The regularizers and the respective parameters, in this case, are decided by the following formulas:
\begin{align}
\!\!r_0^{(u)}(x)=\bm I_{\Delta_{\mathcal{P}^+}}(u), \,\,\,\, r_t^{(u)}(x)=\frac{\sigma_t^{(u)}}{2}\|u-u_t\|^2,\,\,t\geq 1 \label{exp_regula},
\end{align}

\begin{align}
	&\sigma_1^{(u)} = \sigma^{(u)} \sqrt{h_1^{(u)}}, \sigma_t^{(u)}=\sigma^{(u)} \left( \sqrt{ h_{1:t}^{(u)} } - \sqrt{h_{1:t-1}^{(u)} }\right),\ t\geq2  \notag,
	\\ 
	&\text{with} \quad h_t^{(u)}=\|F_t-\p F_t\|^2. \notag 
\end{align}

For the  predictions of the \emph{experts' performance} $\tilde{F}_t$, at each time step, we will use expert's performance of the previous time step\footnote{We note that the motivation for a large corpus of optimistic learning works stems from the fact that in many cases the cost functions are changing slowly \cite{mohri-aistats2016, sridharan-nips2013}. While the motivation in this work has been  different until this section (availability of rec-sys), the optimism in the meta-learner has the same scope with those initial works.}:
\begin{align}
    \p F_t &=\left( \p f_t^{(p)}, \ p \in \mathcal{P}^+\right) \triangleq  \left( f_{t-1}(y_{t-1}^{(p)}),\ p \in \mathcal{P}^+\right). \label{eq:exp_pred}
\end{align}

This type of meta-optimism comes for free since, in practice, we expect the predictors to have consistent accuracy across contiguous  slots; either accurately due to recently trained model, or poorly due to e.g., distributional shift (see \cite{yang_recommender} and references therein) hence we can use the previous function. As for the pessimistic expert, its caching decisions do not vary much in consecutive slots (due to using strongly convex regularizers in updating the decisions).

The execution of the policy is summarized in Algorithm \ref{alg3}; and we also include the step-by-step visualization in Fig. \ref{fig:xc} so as to facilitate the reader. Namely, we see the sequence of steps where: (1) The Predictors output $\{\tilde{c}_t^{(p)}, p\in\mathcal{P}\}$; (2) The optimistic experts optimize for the predictions, whereas the pessimistic expert performs an FTRL step; (3) The experts' proposals are combined via the meta-learner weights; (4) The performance of experts' proposals is calculated through the revealed request; (5) The meta-learner optimistically sets the next-slot experts' performance to be the same as the current one; (6) \& (7) The meta-learner performs an OFTRL step to calculate and set the weights for the next-slot. And the process repeats for the next slot.

\begin{table*}[t]
	\centering
	\caption{Online caching policies with \emph{adversarial} guarantees: a summary of the contributions and comparison with literature. 
		\\
	}
	\label{tab:summary}
	\begin{tabularx}{0.99\textwidth}{|m{1.2cm}||m{6.8cm}|m{2.7cm}|m{2.6cm}|Y| }
		\hline
		\multirow{2}{*}{\textbf{Algorithm}} &   \multirow{2}{*}{ \textbf{Model and Conditions}} & \multicolumn{2}{c|}{\textbf{Guarantees ($R_T, V_T\leq$ )}} &   \multirow{2}{*}{\makecell{\textbf{Adaptive} \\\textbf{Learning}}}
		\\ \cline{3-4} 
		&  & \textbf{Best case} & \textbf{Worst case}  & \\ \hline
		1 ($\pi_{obc}$) &
		$\bullet$ Bipartite network $\bullet$ Coded files $\bullet$  Predictions & \qquad $0$
		& \qquad $O\left(\sqrt{T}\right)$& \checkmark
		\\ \hline 
		2 ($\pi_{ec}$) &  
		$\bullet$ Bipartite  $\bullet$ Coded files $\bullet$  Predictions  $\bullet$ Budget constr. 
		&$O\left(\kappa\sqrt{T}\right), O\left(\kappa T^{\frac{3}{4}}\right)$ & $O\left(\!\frac{\kappa}{2}\sqrt{T}\!\right),  O\left(\!\frac{\kappa}{2} T^{\frac{3}{4}} \!\right)$& \checkmark
		\\ \hline 
		3 ($\pi_{xc}$) &
		$\bullet$ Bipartite network $\bullet$ Coded files $\bullet$ $P
		\geq 1$ predictors
		&\!\! $\!\!A_P\!\triangleq\! \min_{p \in \mathcal{P}^+}\! \left\{\! R_{T}^{(p)}\!\right\}\!\!\!\! $ & $2w\sqrt{(P+1)T} \!+\! A_P$ & \checkmark
		\\ \hline
		\cite{stratis-2020} &
		$\bullet$ Single cache $\bullet$ Coded \& uncoded files
		&  \multicolumn{2}{c|}{$
			O\left(\sqrt{T}\right)$ } & --\\ \hline
		\cite{paschos-infocom19} &
		$\bullet$ Bipartite network $\bullet$ Coded files
		&  \multicolumn{2}{c|}{$
			O\left(\sqrt{T}\right)$ } & --
		\\ \hline 
		\cite{abhishek-sigm20} &
		$\bullet$ Bipartite network $\bullet$ Coded (single cache) \& uncoded
		&  \multicolumn{2}{c|}{$
			O\left(\sqrt{T}\right)$ } & --
		\\ \hline
		\cite{Li-online-2021} &
		$\bullet$ General graph network $\bullet$ Coded \& uncoded files
		&  \multicolumn{2}{c|}{$
			O\left(\sqrt{T}\right)$ } & --
		\\ \hline
		\cite{abhishek_nurips} &
		$\bullet$ Bipartite network
		$\bullet$ Coded \& uncoded files
		&  \multicolumn{2}{c|}{$
			O\left(\sqrt{T}\right)$ } & --
		\\ \hline
	\end{tabularx}
\end{table*}

The following theorem bounds the regret, defined as: 
\[
R^{(xc)}_T \!=  \sum_{t=1}^T \! {c_t}^\top  (y^\star - y_t),
\]
of the proposed meta-learning policy.
\begin{mdframed}
\begin{theorem}\label{th:regret3}
Algorithm 3 ensures the regret bound: 
\begin{align}
    R^{(xc)}_T \! &\leq 2 \sqrt{ \sum_{t=1}^T\|F_t- {F}_{t-1}\|^2} + \min_{p \in \mathcal{P}^+} \left\{ R_{T}^{(p)} \right\} \notag
    \\
    &\leq 2w\sqrt{(P+1)T} + \min_{p \in \mathcal{P}^+} \left\{ R_{T}^{(p)} \right\} \notag.
\end{align}
\end{theorem}
\end{mdframed}

\begin{proof}
We first relate the regret of the combined caching decisions to that of the experts. Then, we can re-use the result of Theorem \ref{th:regret1}:
\begin{align} \!
    R^{(xc)}_T \! &=  \! \sum_{t=1}^T \!\left( {c_t}^\top{y^\star} \! - \! {c_t}^\top\sum_{p \in \mathcal{P}^+} u^{(p)}_{t} \ y^{(p)}_{t}\right) = \sum_{t=1}^ T {c_t}\!^\top\!{y^\star}\! - \! {F_t}\!^\top\!{u_t} \notag 
    \\
    &= \sum_{t=1}^ T {c_t}^\top{y^\star} -  {F_t}^\top{u^\star} + {F_t}^\top{u^\star} - {F_t}^\top{u_t}  \notag 
    \\
    &=R_T^{(u)} + \min_{p \in \mathcal{P}^+} \left\{ R_{T}^{(p)} \right\} \label{eq:comp_reg}
\end{align}
where $R_T^{(u)}$ is the regret for the \emph{weights $u$}: $R_T^{(u)} = \sum_{t=1}^T  {F_t}^\top{u^\star} - {F_t}^\top{u_t} $. Note that \eqref{eq:comp_reg} holds because $u^\star \!\! =\! \argmax_u\ {F_{1:t}}\!\!^\top \! {u} \! = {e^k}$,
$k = \argmax_p f_t(y_t^{(p)})$
and $e^k$ is standard basis vector. Thus, we have:
\begin{eqnarray}
F_{1:t}^\top\ u^\star = \max_{p \in \mathcal{P}^+}\bigg\{ \sum_{t=1}^T f_t (y^{(p)}_t) \bigg\}.
\end{eqnarray}
We use the result of Theorem $1$ to bound $R_T^{(u)}$:
\begin{eqnarray}
    R_T^{(u)} \leq \sqrt{2}\Delta_{\mathcal{P}^+} \sqrt{ \sum_{t=1}^T\|F_t- \tilde{F}_t\|^2} \leq 2 w\sqrt{(P+1)T},
\end{eqnarray}
where the last inequality follows from the simplex diameter ($\sqrt{2}$) and the fact that:
\[
\|F_t - \tilde{F}_t\|^2 \leq \sum_{p \in \mathcal{P}^+}  \|(f_t(y_t^{(p)}) - f_{t-1}(y_{t-1}^{(p)})) \|^2 \leq | \mathcal{P}^+ | w^2
\]
i.e., we predicted a miss or a hit, whichever happened at $t-1$, but the opposite happens at $t$. Substituting in \eqref{eq:comp_reg} gives the bound.
\end{proof}

\textbf{Discussion}. A key observation in Theorem \ref{th:regret3} is that its bound contains the regret of the best optimistic expert. This yields very improved bounds for $R_T^{(xc)}$ whenever there is an optimistic expert that achieves negative regret. For example, if an expert can achieve $R_T\! =\! \Theta(-T)$, e.g., because it has a very accurate rec-sys (its recommendations are most-often followed), then the overall regret is $R_T^{(xc)}\!=\! O(\sqrt{T}) \!- \Theta(-T)$, which becomes strictly negative for large $T$. Moreover, the $O(\sqrt{T})$ term shrinks with more consistent performance of the experts, a condition that is rather expected in practical systems, thus getting us even faster to the regret of the best expert. Nonetheless, since the pessimistic expert exists in the group of experts, the $\min$ term is upper bounded by  $O(\sqrt{T})$ and $R_T^{(xc)}$ will maintain $O(\sqrt{T})$ regardless of the performance of the optimistic experts.

Since Algorithm XC can work with a single optimistic expert, which is the setup handled by $\pi_{obc}$, it is interesting to compare their bounds. In fact, neither of those algorithms is better in all possible scenarios (regarding request patterns; evolution of utility parameters, etc. see also Sec. \ref{sec:evaluation}) than the other; and their relative performance ranking (in terms of utility) depends on the specific problem instance. For example, under worst-case predictions, we have that:\footnote{The pessimistic expert's regret is bounded by $2w\sqrt{CT}$ for the single cache setup (i.e., the diameter $D_\mathcal{Y} = \sqrt{2C}$).} 
\begin{align}
R_T^{(xc)} &\leq 2w \sqrt{(P+1)T} + 2w\sqrt{CT} \notag \\
&\leq 2(\sqrt{2}+\sqrt{C})w\sqrt{T} \qquad\qquad (P=1)
\label{eq:xc_reg_single},
\end{align}
which can be actually better than policy $\pi_{obc}$'s worst-case prediction bound\footnote{Note that $\pi_{obc}$'s worst-case bound is $2w\sqrt{2CT}$ for the single cache setup discussed here.} for practical values of the constants $C$. 

On the other hand, in cases where inaccurate predictions occur for a certain fraction of the steps $\lceil \alpha T \rceil, 0<\alpha<1$ the $\min$ term in the $R_T^{(xc)}$ bound might evaluate to the pessimistic expert's regret since the optimistic experts can suffer linear regrets\footnote{This happens, e.g., when the pessimistic expert achieves a hit on the steps $[\alpha T]$.} $R^{(p)}_T \leq O(\alpha T)$. For $\pi_{obc}$, the regret will be of the form 
\begin{align}
R_T &\leq 2\sqrt{C} \sqrt{\sum_{t\in \left[ \lceil \alpha T \rceil \right]}\|\tilde c_t - c_t\|} \notag \\
&\leq 2\sqrt{2C}w\sqrt{\lceil \alpha T \rceil} \label{eq:obc_reg_alpha}.
\end{align}
For example, for $\alpha\leq 1/2$, The upper bound in \eqref{eq:obc_reg_alpha} is tighter than that in \eqref{eq:xc_reg_single} for all ${C, w}$. Hence, $\pi_{obc}$'s regret can be smaller if in the described case. Overall, the choice between $\pi_{xc}$ and $\pi_{obc}$ in the case of a single predictor depends on the request sequence and the number of steps where predictions fail. Section \ref{sec:evaluation} demonstrates these cases using various scenarios.

\begin{figure*}
	\centering
	\includegraphics[width=0.95\textwidth]{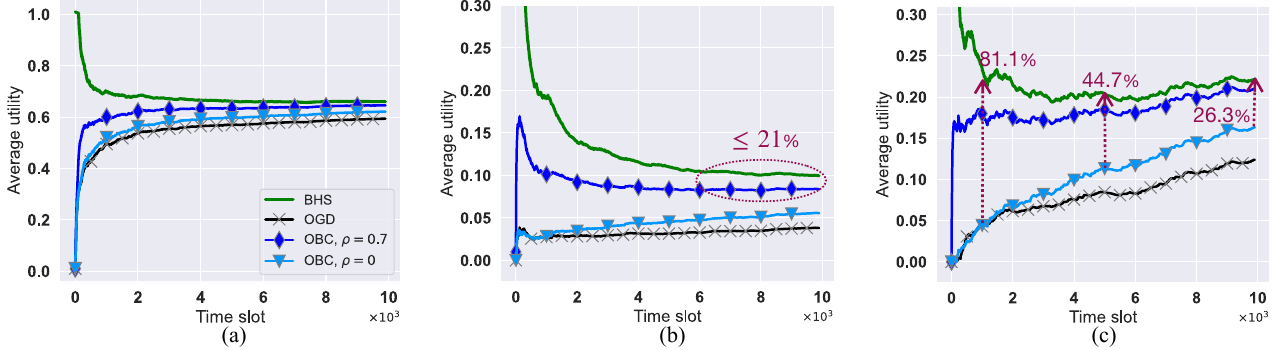}
	\caption{Utility in the single cache model with one rec-sys of different recommendation quality levels (i.e., $\rho$) in (a) Zipf requests with $\zeta=1.1$, (b) YouTube request traces, (c) MovieLens request traces.}
	\label{fig:single_cache_nonexp}
\end{figure*}

\begin{figure*}
	\centering
	\includegraphics[width=0.93\textwidth]{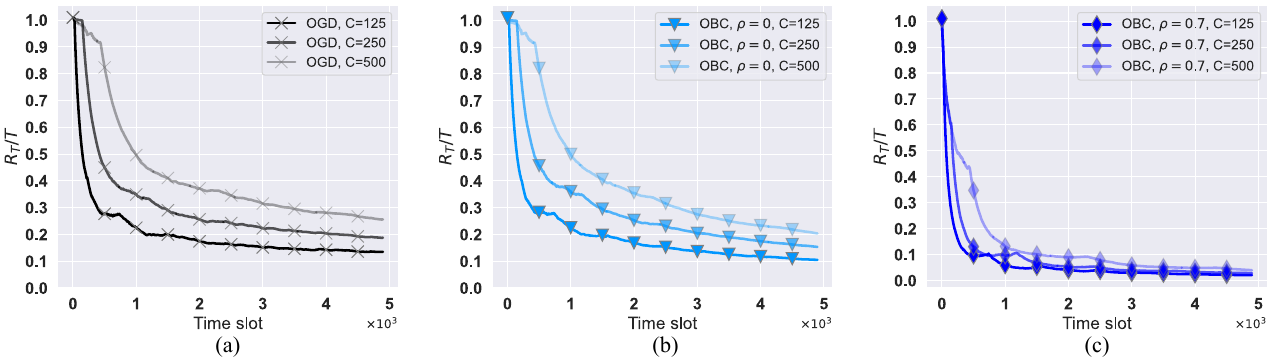}
	\caption{Regret over time in the single cache model with different         values of the cache capacity for (a) $\pi_{ogd}$, (b) $\pi_{obc}$ with $\rho = 0$, (c) $\pi_{obc}$ with $\rho = 0.7$.}
	\label{fig:plot-C}
\end{figure*}

\begin{figure*}
	\centering
	\includegraphics[width=0.99\textwidth]{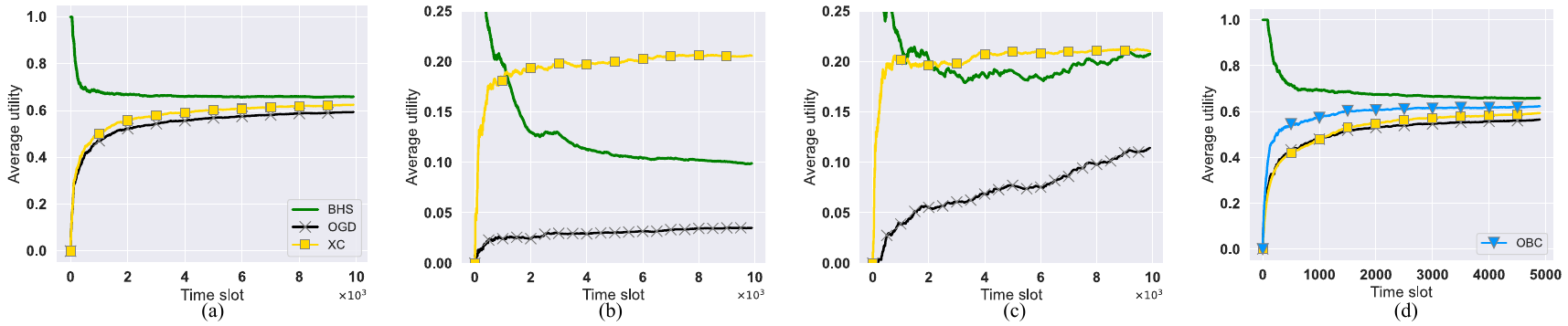}
	\caption{Utility in the single cache model with two rec-sys of recommendation qualities $\rho = 2\%$ and $\rho = 20\%$, each modeled as an expert within XC, in (a) Zipf requests with $\zeta=1.1$, (b) YouTube request traces, (c) MovieLens request traces. (d) A comparison between $\pi_{obc}$ and $\pi_{xc}$ using one rec-sys of an alternating recommendation quality.}
	\label{fig:single_cache}
\end{figure*}

Regarding the extension to caching networks (more than one cache), note that sets $\mathcal{Y}$ and $\mathcal{Z}$ are convex by definition. Also, the set defined by the connectivity constraints \eqref{eq:opt1b} is convex (linear constraint in the variable $x$). Recalling that the Cartesian product and the intersection of convex sets is convex, we conclude that the constraint set $\mathcal{X}$ defined in \eqref{eq:const_set}, from which the joint caching-routing variable $z$ is selected, remains convex. Finally, as demonstrated earlier, the meta-learner takes the convex combination of the experts proposals, and since each expert proposes a caching-routing configuration $z_t^{(p)} \in \mathcal{X}$, the meta caching-routing policy:
\begin{align}
    z_{t+1} = \sum_{p \in \mathcal{P}^+} u^{(p)}_{t+1} \ z^{(p)}_{t+1}, \quad \quad u^{(p)}_{t+1} \in \Delta_{\mathcal{P}^+} \label{eq:comp_action_net}
\end{align}
remains a valid one (i.e., $z_t \in \mathcal{X}, \forall t$). Therefore, indeed, the ideas proposed in this section can be readily applied to bipartite caching networks.

Finally, since we have now presented our algorithms, let us summarize in Table \ref{tab:summary} their main features and revisit how they compare with the state-of-the-art results. For Alg. \ref{alg1} - \ref{alg2}, the best case refers to the scenario where the request predictions are perfect $\tilde{c}_t = c_t, \forall t$; and the worst case to the scenario where predictions are furthest from the truth $\tilde{c}_t = \arg\max_{c}\|c - c_t\|, \forall t$. The dependence of the constant factors of the regret bound, denoted with $\kappa$ to facilitate presentation, was made explicit for Alg. \ref{alg2} where we saw that these constants shrink with the predictions' accuracy; and the same holds for Alg. \ref{alg1}. For Alg. \ref{alg3}, the best case refers to the scenario where the experts' predictions are perfect $\tilde{F}_t = F_t, \forall t$; while the worst case arises when $\tilde{F}_t = \arg\max_{F}\|F - F_t\|, \forall t$. Algorithms that do not leverage predictions in regret analysis (all prior work in caching\footnote{We note the exception of \cite{Lykouris-ML} which considers \emph{ML advice} in the paging problem (single cache, uncoded), but its bounds are defined w.r.t. the cache size and quantified in terms of competitive ratio -- a metric that is not comparable to regret, see discussion in \cite{pmlr-v30-Andrew13}.}) have the best and worst case columns merged. We also distinguish between adaptive and static learning rates. While some prior works do employ time-adaptive learning (dynamic steps), as e.g., in \cite{paschos-infocom19}, none of them adapts the rates (or, equivalently the regularization) to the observed gradients$\{c_t\}_t$ and/or prediction errors, as we propose here, but instead use the Lipschitz constant $w$, where $\|c_t\|\leq w, \forall t$. This leads to looser bounds in most practical cases \cite{mcmahan-survey17} and, of course, does not allow to benefit from the availability of predictions.

\section{Performance evaluation}\label{sec:evaluation}

\begin{figure*}[t]
\centering
	\includegraphics[width=0.91235\textwidth]{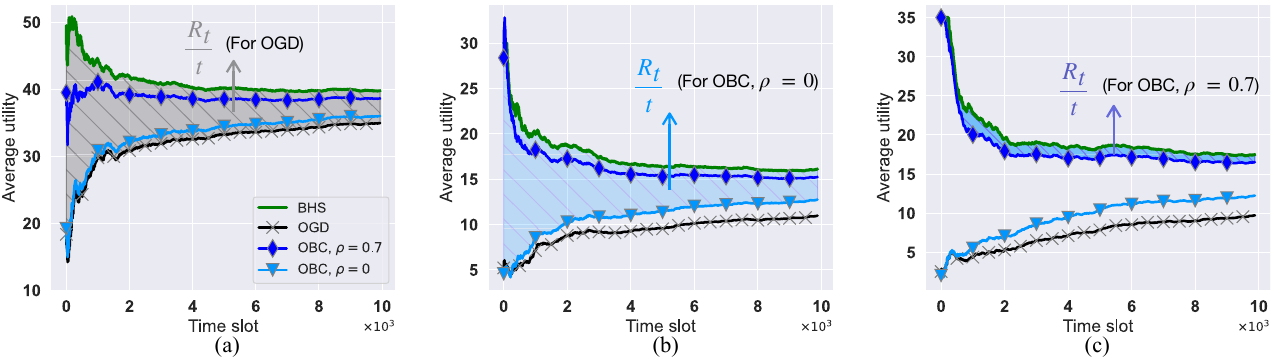}
	\caption{Attained utility in the bipartite model under different recommendation quality levels in (a) Zipf requests with $\zeta=1.1$, (b) YouTube request traces, (c) MovieLens request trace.}
	\label{fig:bi_unconst}
\end{figure*}

We evaluate $\pi_{obc}$, $\pi_{oec}$ and $\pi_{xc}$ under different request patterns and predictions modes; and we benchmark them against $x^\star$ and the OGD policy \cite{paschos-infocom19} that outperforms other state-of-the-art policies \cite{giovanidis-mLRU, leonardi-implicit}. We observe that when reasonable predictions are available, the proposed policies have an advantage, and under noisy predictions, they still reduce the regret at the same rate with OGD, as proven in the Theorems. First, we compare $\pi_{obc}$ and $\pi_{xc}$ against OGD \cite{paschos-infocom19} in the single cache case. We then study $\pi_{obc}$ for the bipartite model and $\pi_{oec}$ with the presence of budget constraints. We consider three requests scenarios, stationary Zipf requests (with parameter $\zeta=1.1$) and two actual request traces: YouTube (YT) \cite{zink2008watch} and MovieLens (ML) \cite{mlds}. For predictions, we assume that at each time step, the user follows the recommendation with probability $\rho$ (unknown to the caching system), and we experiment with different $\rho$ values. The full codebase for the proposed policies and experiments is available via GitHub \cite{code}.

\begin{figure*}[t]
	\centering
	\includegraphics[width=0.939\textwidth]{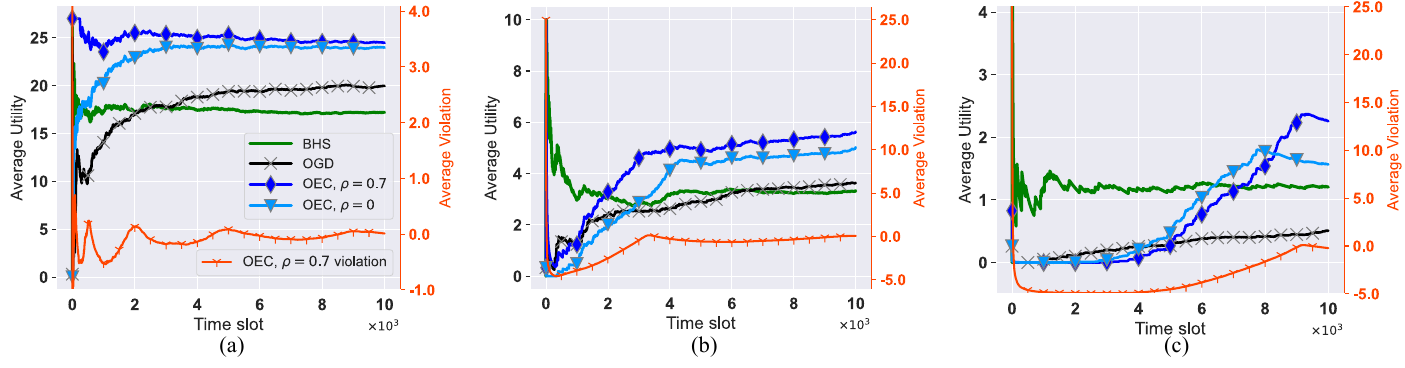}
	\caption{Utility and budget utilization with (a): Zipf requests with $\zeta=1.1$ and (b): YouTube traces (c): MovieLens traces.}
	\label{fig:bi_const_st}
\end{figure*}

\textbf{Single Cache Scenarios}.
We set $w=1$ to study the cache hit rate scenario, use a library size of $N=10^4$ files, and cache capacity of $C=100$ files. Figures \ref{fig:single_cache_nonexp}.a-c. depict the attained average utility, $\frac{1}{t}\sum_{i=1}^t f_i(x_i)$, for each policy and the Best in Hindsight (BHS) cache configuration \emph{until that slot}, i.e., we find the best in hindsight\footnote{Unlike \cite{paschos-infocom19} that calculates $x^\star$ for the largest $t$, we use $x_t^\star$ for each $t$. Thus, the gap between any policy and BHS at $t$ is the policy's average regret $R_t/t$.} for each $t$. Note that BHS always achieves utility $1$ initially (first requests to fill the cache). Thus, we cut the y-axis for better presentation in Figures \ref{fig:single_cache_nonexp}.b,c. It can be seen that the accurate predictions (i.e., when users follow the recommendations $70\%$ of the time) pushes the performance of our online caching towards BHS. For example, in \ref{fig:single_cache_nonexp}.b, the utility gap is at most $21\%$ after $t=6k$. At the same time, even when users do not follow the recommendation at all, we still maintain a diminishing regret; the gap in \ref{fig:single_cache_nonexp}.c goes from $81.1\%$ at $t=1k$, to $26.3\%$ at $t=10k$. In Fig. \ref{fig:plot-C}, we study the effect of increasing the cache size. Expectedly, the regret increases since the diameter of the decision set also increases. However, with higher prediction quality, this effect is minimized since the regret is proportional to a shrinking error term.

For multiple predictors, we use $\pi_{xc}$. 
In figures \ref{fig:single_cache}.a-c, we evaluate $\pi_{xc}$ with 3 experts: an FTRL expert, and two other optimsitic experts. The first optimistic expert is endowed with a predictor (e.g., a recommendation system) that gets followed with probability $\rho=2\%$. For the other it is  $\rho=20\%$. As shown in the plots, $\pi_{xc}$ achieves negative regret on the traces (it outperforms the BHS policy) and converges to the performance of the best expert ($0.20$ utility). This is because in more spread distributions and real request traces, predicting the next request provides great advantage for policies that modify the cache online over the fixed BHS. In the stationary Zipf request pattern, the optimal cache is for the files with top probabilities, which are easily captured by BHS. Thus, BHS policy performs the best. In Fig. \ref{fig:single_cache}.d we show the advantages of $\pi_{obc}$ compared to $\pi_{xc}$ \emph{with two experts}: an FTRL expert and a recommendation-based expert. $\rho$ alternates between $100\%$ and $0\%$ (i.e., requesting the file recommended at a time step $t$, and any other file at $t+1$, and so on). Here, $\pi_{obc}$ outperforms $\pi_{xc}$ since the alternating prediction accuracy induces frequent switching between the two experts in $\pi_{xc}$: the performance of the optimistic expert alternate between $0$ and $1$, while that of pessimistic expert is in the range ($0.55, 0.65$). Hence, $\pi_{xc}$ is inclined to place some weight on the prediction expert at one step, only to retract and suffer a greater loss at the following one had it stayed with the full weight on the FTRL expert. Due to the additional regret caused by such frequent switching, $\pi_{obc}$'s regret is $54.8\%$ of $\pi_{xc}$'s. $\pi_{obc}$ also achieves $38.2\%$ of $\pi_{ogd}$'s regret. 



\textbf{Bipartite Networks}. We consider next a bipartite graph with $3$ caches and $4$ user locations, where the first two locations are connected with caches $1$ and $2$, and the rest are connected to caches $2$ and $3$. The utility vector is $w_n=(1, 2, 100), \forall i,j$, thus an efficient policy places popular files on cache $3$. This is the setup used in \cite{paschos-infocom19} that we adopt here to make a fair comparison. For the zipf scenario, we consider a library of $N=1000$ files and $C=100$. For the traces scenario, files with at least $10$ requests are considered, forming a library of $N=456$ files for the YouTube dataset, and we set $C=50$, and $N=1152$ for the ML dataset, and we increase $C=100$. The location of each request is selected uniformly at random. Similar to the single-cache case, we plot the average utility of the online policies and the best static configuration \emph{until each $t$}. Recall that the area between a policy and BHS is the average regret of that policy. To avoid clutter, we shade this area for OGD in the first sub-figure, as an example, and for OBC in the next two.

In Fig. \ref{fig:bi_unconst}.a, the effect of good predictions is evident as OBC maintains utility within $5.5\%$ of BHS's utility after $t=2.5k$. Even when the recommendations are not followed, OBC preserves the sublinear regret, achieving a gap of $30.4\%$ and $8.5\%$ for  $t=1k$ and $t=10k$, respectively. Akin patterns appear in the traces scenarios. Namely the similarity between OBC with good predictions and BHS, and the improvement in OBC utility despite the recommendations quality. We also note lower utility scores across all policies due to the more spread request.  

Next, we consider the case of budget constraint and evaluate $\pi_{oec}$ for the zipf case, Fig. \ref{fig:bi_const_st}.a,  and the two traces: Fig. \ref{fig:bi_const_st}.b, c. The prices at each slot are generated uniformly at random in the normalized range $[0, 1]$, and the available budget is generated randomly $b_t = \mathcal{N}(0.5, 0.05)  \times  10$ i.e., enough for approximately $10$ files. Such tight budgets magnify the role of dual variables and allow us to test the constraint satisfaction. The benchmark $x^\star$ is computed once for the \emph{full time horizon}, and its utility is plotted for each $t$. In both scenarios, we note that the constraint violation for all policies is approximately similar, fluctuating during the first few slots and then stabilizing at zero. Hence, we plot it for one case.

\begin{figure*}
	\centering
	\includegraphics[width=0.85\textwidth]{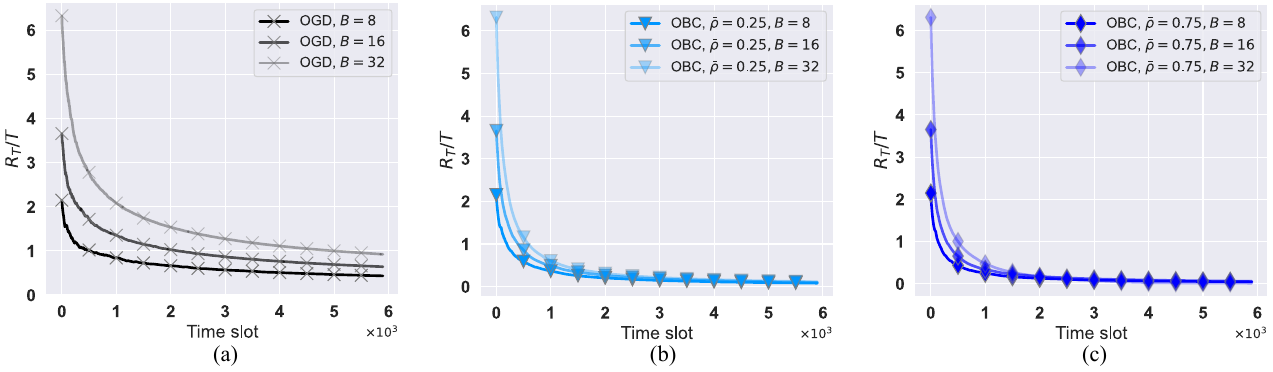}
	\caption{Average regret over time for the single cache, generalized batched requests model for (a) $\pi_{ogd}$, (b) $\pi_{obc}$ with $\bar\rho =         0.25$, and (c) $\pi_{obc}$ with $\bar\rho = 0.75$.}
	\label{fig:batched}
\end{figure*}

Concluding, we find that $\pi_{oec}$ can even outperform the benchmark. This is because the actual request patterns in the traces are not actually adversarial. Also, unlike the benchmark policy, $\pi_{oec}$ is allowed to violate the budget at some time slots, provided that the constraints are eventually satisfied, which occurs either due to strict satisfaction or due to having an ample subsidy at some slots. For example, in the first scenario (Fig. \ref{fig:bi_const_st}.a), the good predictions enable OEC to outperform $x^\star$ by $42.2\%$ after observing all requests ($T=5K$). OGD, and OEC with noisy predictions attain utility units improvement of $16.1\%$, $39.3\%$, respectively, over the BHS. We note in Fig.\ref{fig:bi_const_st}.b, c the delay in learning due to initially over-satisfied budget constrain. This is a transient effect as the dual variables approach their optimal value. Eventually, the bounds on regret are satisfied. We stress that the algorithms scale for very large libraries $\mathcal N$ the only bottleneck in the simulations is finding $x^\star$, which involves the horizon $T$; this is not required in real systems.

\textbf{Computational Complexity.} Note that the OFTRL update (e.g., \eqref{proxy-step}) is either a linear program in case $r_t(\cdot) = 0, \forall t$, or a quadratic program otherwise. In the former case, the solution is finding the top $C$ most requested files. This can be done through a simple ordering operation whose worst-case complexity is linear in the dimension of the decision variable (e.g., $O(N\log(N)$ for each cache). If, however, we have regularization terms, then the resulting quadratic program can be solved in closed form in $\mathbb{R}$. Then, for the projection to the feasible set $\mathcal{X}$, we can use specialized projection algorithms for the capped simplex (the capacity constraints for each cache) whose worst-case complexity is still polynomial in the dimension (i.e., $O(N^2)$)\cite{proj_cs}. 

The above discussion is a worst-case one. In practice, we use the fact that we repeatedly solve similar optimization problems at each time step (the update step differs by adding one linear term and one quadratic term). Thus, the solution in a step can be used as an initial point for the following one. We note that the experimental running times are significantly less than the worst-case ones, as can be seen in Fig. \ref{fig:TC-1} and Fig. \ref{fig:TC-2}. In these figures, these simulations were done using \texttt{CVXPY} $1.2$ package with \texttt{Python} $3.10$ running on an Apple M1 Pro Chip and 16GB of RAM. The differences between OGD and the optimistic one is due to the difference in the update problem structure (lazy vs. greedy projection).

\begin{figure}
	\centering
	\includegraphics[width=0.45\textwidth]{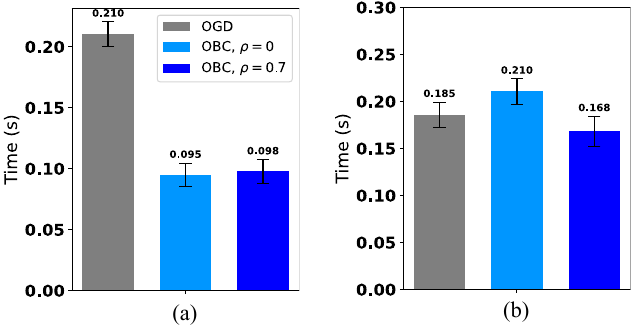}
	\caption{Average time consumed per decision step in the bipartite            network configuration and the ML Dataset ($N=1152$, $C=100$, $|\mathcal{I}|=4$, $|\mathcal{J}|=3$) with (a) pre-reserved storage ($\pi_{obc}$), and      (b) elastic storage ($\pi_{oec}$).}
	\label{fig:TC-1}
\end{figure}

\begin{figure}
	\centering
\includegraphics[width=0.45\textwidth]{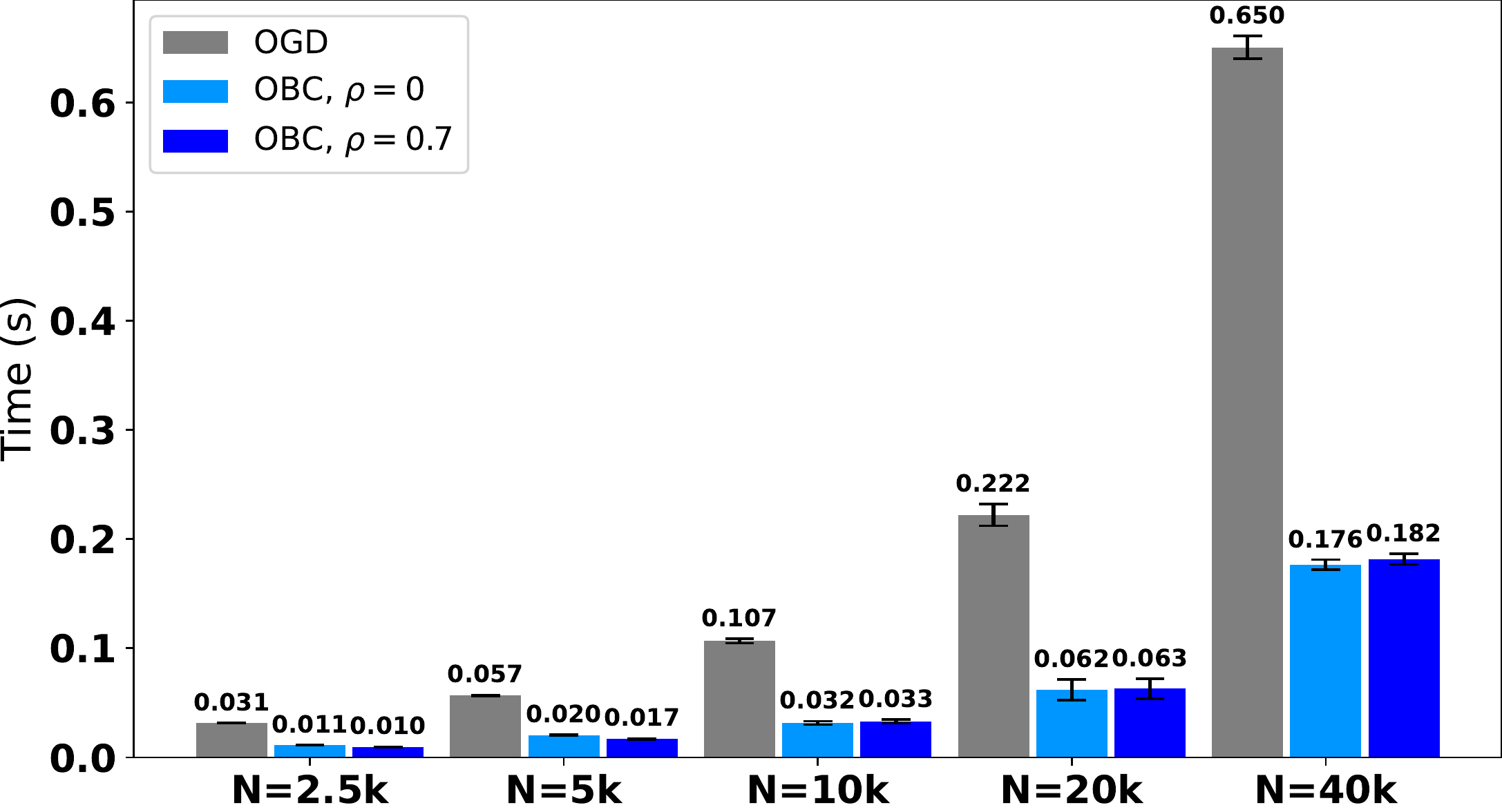}
	\caption{Average time consumed per decision step in the single cache configuration and stationary requests with different library sizes $N$. The sample size is $T=10k$ }
	\label{fig:TC-2}
\end{figure}

\textbf{Batched Requests.} As mentioned in the discussion of the system model, as well as that of Theorem 1, our assumption on the request model (one request per time slot) is for technical ease of analysis. The main feature of the proposed policies, which is having $R_T \propto O({\sqrt{\sumT\|c_t - \p c_t\|}})$, remains valid when the request model is batched (i.e., processing $B$ requests per time slot). However, the upper bound (not necessarily the regret) will be scaled accordingly since the diameter of the decision set will now increase. Namely, following the same steps in the proof of Theorem. 1 (after \eqref{mohri-rt}), instead of $D_{\mathcal{X}} = 2(JC+1)$, we would have $D_{\mathcal{X}} = 2(JC+B)$. Fig. \ref{fig:batched} shows experimentally how the batch size affects the regret. In this experiment, we introduce a new parameter, $\bar\rho$, which denotes the percentage of requests correctly predicted out of the total $B$. We note in these experiments that with better predictions, the effect is amortized since the term $\sqrt{\sumT\|c_t - \p c_t\|}$ remains small.

\section{Conclusions}\label{sec:conclusions}

The problem of online caching is timely with applications that extend beyond content delivery to edge computing and in fact to any dynamic placement problem with Knapsack-type constraints. This work proposes a new suite of caching policies that leverage predictions obtained from content recommendations, and possibly other forecasters, to minimize the caching regret w.r.t an ideal (yet unknown) benchmark cache configuration. As recommender systems permeate online content viewing platforms, such policies can play an essential role in optimizing caching efficacy. We identified and built upon this new connection between caching and recommender systems. The proposed algorithmic framework is scalable and robust to the quality of recommendations and the possible variations of network state and the request sequences, which can even be decided by an adversary. The achieved bounds improve upon the previously known caching regret performance, see \cite{paschos-infocom19,Li-online-2021, stratis-2020, abhishek-sigm20} and references therein. Finally, we believe this work opens new research directions both in terms of caching, e.g., pursuing the design of optimistic policies for uncoded caching; and in terms of resource scheduling in pertinent network and mobile computing problems using untrusted sources of optimism, i.e., predictors of unknown or varying accuracy.

\bibliography{references_short}
\bibliographystyle{IEEEtran}
\end{document}